\documentclass[sigconf]{preprint}

\makeatletter 
\gdef\acmConference@shortname{}   
\gdef\acmConference@name{}        
\gdef\acmConference@date{}        
\gdef\acmConference@venue{}       
\gdef\acmConference@editors{}     
\makeatother 

\usepackage{balance} 
\usepackage{enumitem}

\setcopyright{none}

\acmSubmissionID{1665}

\title{Beyond Outcome-Based Imperfect-Recall: Higher-Resolution Abstractions for Imperfect-Information Games}

\author{Yanchang FU}
\affiliation{
  \department{School of Artificial Intelligence} 
  \institution{University of Chinese Academy of Sciences}
  \city{Beijing}
  \country{China}
}
\additionalaffiliation{
  \department{Institute of Automation}
  \institution{Chinese Academy of Sciences}
  \city{Beijing}
  \country{China}
}
\email{fuyanchang2020@ia.ac.cn}

\author{Qiyue YIN}
\affiliation{
  \department{Institute of Automation}
  \institution{Chinese Academy of Sciences}
  \city{Beijing}
  \country{China}
}
\email{qyyin@nlpr.ia.ac.cn}

\author{Shengda LIU}
\affiliation{
  \department{Institute of Automation}
  \institution{Chinese Academy of Sciences}
  \city{Beijing}
  \country{China}
}
\email{shengda.liu@ia.ac.cn}

\author{Pei XU}
\affiliation{
  \department{Institute of Automation}
  \institution{Chinese Academy of Sciences}
  \city{Beijing}
  \country{China}
}
\email{xupei2018@ia.ac.cn}

\author{Kaiqi HUANG}
\affiliation{
  \department{Institute of Automation}
  \institution{Chinese Academy of Sciences}
  \city{Beijing}
  \country{China}
}
\email{kqhuang@nlpr.ia.ac.cn}

\begin{abstract}

Hand abstraction is crucial for scaling imperfect-information games (IIGs) such as Texas Hold'em, yet progress is limited by the lack of a formal task model and by evaluations that require resource-intensive strategy solving. We introduce \textbf{signal observation ordered games (SOOGs)}, a subclass of IIGs tailored to hold'em-style games that cleanly separates signal from player action sequences, providing a precise mathematical foundation for hand abstraction. Within this framework, we define a resolution bound—an information-theoretic upper bound on achievable performance under a given signal abstraction. Using the bound, we show that mainstream outcome-based imperfect-recall algorithms suffer substantial losses by arbitrarily discarding historical information; we formalize this behavior via \textbf{potential-aware outcome Isomorphism (PAOI)} and prove that PAOI characterizes their resolution bound. To overcome this limitation, we propose \textbf{full-recall outcome isomorphism (FROI)}, which integrates historical information to raise the bound and improve policy quality. Experiments on hold'em-style benchmarks confirm that FROI consistently outperforms outcome-based imperfect-recall baselines. Our results provide a unified formal treatment of hand abstraction and practical guidance for designing higher-resolution abstractions in IIGs.

\end{abstract}


\newcommand{\BibTeX}{\rm B\kern-.05em{\sc i\kern-.025em b}\kern-.08em\TeX}

\begin{document}


\pagestyle{fancy}
\fancyhead{}


\maketitle 

\section{Introduction}

In recent years, artificial intelligence (AI) for \textbf{heads-up no-limit hold'em} (HUNL)—a two-player poker variant and classic testbed for \textbf{imperfect-information games} (IIGs, where players lack full knowledge of the current state during playing)—has achieved groundbreaking milestones. AI systems including DeepStack, Libratus, and Pluribus~\cite{moravvcik2017deepstack, brown2018superhuman, brown2019superhuman} have all defeated top human professionals, highlighting AI's potential in mastering complex strategic reasoning.

Hand abstraction has emerged as a game-simplification technique to address the challenges of building AI for large-scale IIGs, particularly HUNL. This need arises because state-of-the-art (SOTA) IIG solvers rely on \textbf{counterfactual regret minimization (CFR)}~\cite{zinkevich2007regret} and its variants~\cite{lanctot2009monte, tammelin2014solving, brown2019solving, zhou2020lazy}, whose prohibitive spatial overhead far exceeds current computational resources. Among existing hand abstraction methods, \textbf{potential-aware abstraction with earth mover's distance (PAAEMD)}~\cite{ganzfried2014potential} is the SOTA and has been integrated as a critical component in high-performance HUNL AIs like DeepStack, Libratus, and Pluribus.

However, hand abstraction—primarily an engineering-driven technique—has three key limitations. First, it lacks a formal mathematical model: no prototype exists in established IIG frameworks to characterize its game-simplification logic, leaving no systematic tool to guide the design of hand abstraction algorithms (development relies heavily on empirical trial and error). Second, its evaluation is indirect and resource-intensive: a strategy must first be solved for the abstracted game, then assessed in the original game. This undermines its core goal of simplifying strategy solving—even solving the abstracted game incurs substantial overhead, yet remains a prerequisite for evaluation. Third, mainstream abstraction algorithms (e.g., PAAEMD) fully discard historical game information, causing severe information loss. This inherent flaw caps performance, as critical prior-state context (essential for distinguishing strategically distinct hands) is arbitrarily excluded, compromising the abstraction's ability to preserve the original game's strategic structure.

This paper makes three key contributions with structure as follows to elaborate on these workstreams.

First, to address the lack of refined modeling for hand abstraction, we construct \textbf{signal observation ordered games (SOOGs)}—a subset of IIGs tailored to capture the characteristic components of Texas Hold'em-style games (referred to as \textbf{hold'em games} throughout, including HUNL). Building on SOOGs, we further develop a formal mathematical model for hand abstraction (termed \textbf{signal observation abstraction}). This is detailed in Section~\ref{sec:models}, following essential background and notations in Section~\ref{sec:background}.

Second, to overcome the inefficiency of indirect hand abstraction evaluation (which relies on resource-heavy strategy solving), we propose \textbf{resolution bound}—a rational indicator to directly measure the maximum potential performance of hand abstraction algorithms—and validate its suitability as an evaluation metric via Theorem~\ref{thm:monotonicity} in Section~\ref{sec:evaluation-poi-froi}.

Third, we design two targeted hand abstractions: \textbf{potential-aware outcome isomorphism (PAOI)} and \textbf{full-recall outcome isomorphism (FROI)}. PAOI illustrates flaws in mainstream algorithms and is proven to be their resolution bound, explicitly revealing critical limitations (notably arbitrary historical information discard). FROI, by contrast, demonstrates how integrating historical information mitigates these flaws and improves abstraction quality. Section~\ref{sec:evaluation-poi-froi} further details both methods and verifies PAOI's role as a baseline bound.

These contributions lead to our core conclusion: arbitrary discard of historical information degrades hand abstraction performance, while integrating it preserves the original game's strategic features and boosts performance. Section~\ref{sec:experiment} provides experimental validation, confirming FROI's significant performance superiority over PAOI and reinforcing the value of historical information for advancing hand abstraction techniques.

\subsection{Related Works}

Hand abstraction for hold'em games traces its roots to the foundational work of \citet{shi2001abstraction} and \citet{billings2003approximating}. \citet{gilpin2006competitive} made early attempts at automatic methods in this field, marking a shift from manual hand categorization. Later, \citet{gilpin2007lossless} introduced \textbf{lossless isomorphism (LI)} for hand abstraction—an approach that constructs equivalence classes via card suit rotation. However, LI's low compression rate shifted research focus toward lossy methods based on hand outcome win rates and \textbf{imperfect-recall}~\cite{waugh2009practical}, including expected  hand strength (EHS)~\cite{gilpin2007better}, potential-aware abstractions (PAAs)~\cite{gilpin2007potential}, and their SOTA variant, PAAEMD~\cite{ganzfried2014potential}. Notably, the aforementioned lossy abstraction algorithms all follow the paradigm of arbitrary historical information discard, relying on extreme imperfect-recall; this is a critical limitation we aim to address in this work.

For IIG modeling relevant to hold'em games, \citet{gilpin2007lossless}'s \textbf{games with ordered signals} framework is the closest, yet it suffers from cumbersome phase modeling, limited generalizability, and a lack of formal hand abstraction models. Subsequent IIG abstraction research (e.g., \citet{waugh2009abstraction}; \citet{kroer2016imperfect, kroer2018unified}) focused on abstraction but conflated it with techniques like action abstraction. Works on public states (\citet{johanson2011accelerating}) and public belief states (\citet{vsustr2019monte}; \citet{lisy2017eqilibrium}) first noted the uniqueness of public action sequences in hold'em games (i.e., they can be studied independently), while recent factored-observation stochastic games (FOSG) models~\cite{kovavrik2022rethinking,schmid2023student} further recognized that observations can be modeled independently—and applied to broader scenarios beyond hold'em games. These works advanced IIG modeling but did not involve hand abstraction.

\section{Background} \label{sec:background}
In this section, we first provide an overview of hold'em games, then introduce the core concept of imperfect-information games (IIGs), and finally touch on content related to strategy solving for them.

\subsection{Hold'em Games}

Hold'em games are a class of poker games where players combine private cards (hole) with shared community cards (board) to construct the best possible x-card poker hand. These games involve multiple betting phases, during which players make decisions based on their hands and their observations of opponents' actions. Players can take multiple actions during the game, with the outcome determined either by all but one player conceding (folding) or by a showdown where the remaining players reveal their hands to determine the winner.

\textbf{Heads-up limit hold'em (HULH)} and HUNL are prominent benchmark environments for AI research in IIGs. In both games, two players compete using two private cards and five community cards revealed over four betting phases: Preflop, Flop, Turn, and River. Notably, while the two variants differ in betting structures—HULH features fixed bet sizes in each phase (facilitating more tractable analysis), and HUNL allows players to bet any amount up to their entire stack (resulting in a larger action space and greater strategic complexity) — their rules for dealing cards and forming valid hands are identical. Since this paper focuses on hand abstraction, we will only involve HULH in subsequent sections.

Hold'em games also include other variations that serve as toy games in AI research, such as Kuhn poker~\cite{kuhn1950simplified}, Leduc Hold'em~\cite{southey2005bayes}, and Flop Hold'em~\cite{brown2017safe}. These simplified games reduce the complexity of full-scale poker while preserving core strategic elements, making them valuable for AI experimentation and theoretical analysis. 

\subsection{Imperfect-Information Games}

Imperfect-information games are a standard mathematical model for modeling hold'em games.

\begin{definition}[Imperfect-Information Game (IIG)] \label{def:iig}
An imperfect-information game \(\mathcal{G}\) is described by the following components:

\begin{itemize}[left=0cm]
    \item \(\mathcal{N}_c = \mathcal{N} \cup \{c\}\): A finite set of players. \(\mathcal{N} = \{1, \ldots, N\}\) denotes rational players (the actual game participants), and the special player \(c\) represents random factors (referred to as \textbf{nature} or \textbf{chance}).
    
    \item \(\mathcal{A}\): The set of all possible actions available to players in \(\mathcal{N}_c\) throughout the game.
    
    \item \(H\): A finite set of histories, where each history is a sequence of successive player actions. The empty sequence \(h^o \in H\) is the initial history. The notation \(h \sqsubset h'\) means \(h\) is a predecessor of \(h'\) (and \(h'\) is a successor of \(h\))—formally, there exists \(a_1, a_2, \ldots, a_k \in \mathcal{A}\) such that \(h' = h \cdot a_1 \cdot a_2 \cdot \ldots \cdot a_k\) (i.e., \(h'\) is \(h\) with the actions appended).
    
    \item \(Z \subseteq H\): The set of terminal histories (histories with no successors), where the game terminates.
    
    \item \(\rho: H \setminus Z \mapsto \mathcal{N}_c\): An action assignment function that specifies exactly one player to act at each non-terminal history. It induces a partition \(\{H_1, \ldots, H_N, H_c\}\) of \(H\setminus Z\), where \(H_i = \{ h \in H\setminus Z \mid \rho(h) = i \in \mathcal{N}_c \}\) (the set of histories where player \(i\) acts).
    
    \item \(A: H \setminus Z \mapsto 2^{\mathcal{A}}\): A legal action function that maps each non-terminal history to its set of legal actions. Here, \(2^{\mathcal{A}}\) denotes the power set of \(\mathcal{A}\) (i.e., all subsets of \(\mathcal{A}\)).
    
    \item \(\zeta\): A chance probability function. For each non-terminal history \(h \in H_c\) (where nature acts), \(  \zeta(h, \cdot): A(h) \to [0,1]\) defines a probability distribution over \(A(h)\), satisfying the normalization condition \(\sum_{a \in A(h)}  \zeta(h, a) = 1\).
    
    \item \(u = (u_i)_{i \in \mathcal{N}}\): A tuple of utility functions. For each player \(i \in \mathcal{N}\), \(u_i: Z \to \mathbb{R}\) assigns a real-valued payoff to \(i\) at every terminal history \(z \in Z\) (determining \(i\)'s reward when the game ends at \(z\)).
    
    \item \(\mathcal{I} = (\mathcal{I}_i)_{i \in \mathcal{N}}\): A collection of information partitions. For each \(i \in \mathcal{N}\), \(\mathcal{I}_i\) partitions \(H_i\) into disjoint \textbf{information sets (infosets)}. If \(h, h' \in H_i\) belong to the same infoset \(I \in \mathcal{I}_i\), rational player \(i\) cannot distinguish \(h\) and \(h'\); this implies \(A(h) = A(h')\), so we overload notation as \(A(I) := A(h) = A(h')\).
\end{itemize}
\end{definition}

\subsection{Strategies and Solution Concepts for IIGs}
In an IIG, each player \(i \in \mathcal{N}\) selects a \textbf{(behavioral) strategy} \(\sigma_i\), where \(\Sigma_i\) denotes the set of all such strategies for \(i\). A strategy \(\sigma_i\) describes the probability distribution over actions chosen at each infoset: formally, \(\sigma_i = \{ \sigma_i(I, \cdot) \mid I \in \mathcal{I}_i \}\), where \(\sigma_i(I, \cdot): A(I) \to [0,1]\) satisfies \(\sum_{a \in A(I)} \sigma_i(I, a) = 1\) (normalization) for each infoset \(I \in \mathcal{I}_i\). When all rational players select their strategies, they form a strategy profile \(\sigma = (\sigma_1, \ldots, \sigma_N)\), with the set of all possible profiles denoted by \(\Sigma = \times_{i \in \mathcal{N}} \Sigma_i\).

To compute payoffs under a strategy profile \(\sigma \in \Sigma\), we first define the reaching probabilities of histories. For any history \(h \in H\), let \(\pi_c(h) = \prod\nolimits_{h' \sqsubset h, h' \cdot a \sqsubseteq h, \rho(h') = c} \zeta(h', a)\) denote the contribution of nature to the probability of reaching \(h\) (where \(h' \sqsubseteq h\) extends \(h' \sqsubset h\) to include the case \(h' = h\)); for a player \(i \in \mathcal{N}\), let \(\pi_i^\sigma(h) = \prod\nolimits_{h' \sqsubset h, h' \cdot a \sqsubseteq h, \rho(h') = i} \sigma_i(I[h'], a)\) denote the contribution of \(i\)'s strategy \(\sigma_i\) to reaching \(h\) (with \(I[h] \in \mathcal{I}_i\) being the infoset containing \(h\)). Since nature's choices are independent of players' strategies, \(\pi_c(h)\) is invariant to \(\sigma\) (and thus we omit \(\sigma\) in its notation, avoiding distinction between \(\pi_c^\sigma(h)\) and \(\pi_c(h)\)). The total probability of reaching \(h\) is then \(\pi^\sigma(h) = \prod_{i \in \mathcal{N}_c} \pi_i^\sigma(h)\), and the expected payoff of player \(i\) under \(\sigma\) is \(u_i(\sigma) = \sum_{z \in Z} \pi^\sigma(z) \cdot u_i(z)\).

An IIG has \textbf{perfect-recall} if every player \(i \in \mathcal{N}\) never forgets past actions or observations: for any \(h, h' \in I \in \mathcal{I}_i\), every predecessor \(h'' \sqsubset h\) with \(h'' \in H_i\) has a unique counterpart \(h''' \sqsubset h'\) in the same infoset (i.e., exists \(I' \in \mathcal{I}_i\) s.t., \(h'', h''' \in I'\)), and the actions taken by player \(i\) at \(h''\) and \(h'''\) (to reach \(h\) and \(h'\)) are identical. Under perfect-recall, a strategy profile \(\sigma^* = (\sigma_1^*, \ldots, \sigma_N^*)\) is a \textbf{Nash equilibrium} if for all \(i \in \mathcal{N}\) and \(\sigma_i \in \Sigma_i\), \(u_i(\sigma^*) \geq u_i(\sigma_i, \sigma_{-i}^*)\) (where \(\sigma_{-i}^*\) denotes the strategy profile of all players except \(i\)). Every finite perfect-recall IIG guarantees at least one such equilibrium.

For 2-player zero-sum IIGs (\(N=2\), \(u_1(z) = -u_2(z)\) for all \(z \in Z\)) with perfect-recall, the \textbf{Minimax Theorem} holds: there exists a unique game value \(v\), and strategies \(\sigma_1^* \in \Sigma_1\) and \(\sigma_2^* \in \Sigma_2\) such that \[\max_{\sigma_1 \in \Sigma_1} \min_{\sigma_2 \in \Sigma_2} u_1(\sigma_1, \sigma_2) = \min_{\sigma_2 \in \Sigma_2} \max_{\sigma_1 \in \Sigma_1} u_1(\sigma_1, \sigma_2) = v.\] While multiple Nash equilibria may exist, all yield this unique value \(v\), and any pair of maximin and minimax strategies \((\sigma_1^*, \sigma_2^*)\) forms a Nash equilibrium. \textbf{Exploitability} measures how far a strategy profile \(\sigma = (\sigma_1, \sigma_2)\) is from any Nash equilibrium \(\sigma^* = (\sigma_1^*, \sigma_2^*)\), quantified by the maximum gains each player can achieve by deviating from \(\sigma\)). Its exploitability is defined as:
\[
\epsilon(\sigma) = \frac{
    \overbrace{\left( u_1(\sigma^*) - \min\limits_{\sigma'_2 \in \Sigma_2} u_1(\sigma_1, \sigma'_2)  \right)}^{\epsilon_1(\sigma)} + 
    \overbrace{\left( u_2(\sigma^*) - \min\limits_{\sigma'_1 \in \Sigma_1} u_2(\sigma'_1, \sigma_2)  \right)}^{\epsilon_2(\sigma)}
}{2},
\]
where a profile with \(\epsilon(\sigma) \leq \epsilon\) is called an \(\epsilon\)-Nash equilibrium.


\section{Modeling for Hand Abstraction} \label{sec:models}

In this section, we present a more precise modeling of the hand abstraction task. To this end, we first impose a series of constraints on IIGs to construct a subset of IIGs, namely signal observation ordered games (SOOGs), which is used to model hold'em games. We then model the hand abstraction task using signal observation abstraction within the SOOG framework.

\subsection{Signal Observation Ordered games}

Modeling hold'em games via traditional imperfect-information games has a critical limitation: infosets are treated as indivisible atomic units. In reality, a player's decision state in hold'em games inherently encompasses two independent core dimensions—chance action sequences (hand dealings) and rational players' action sequences. These are independent in that subsets can be defined for each separately, with any combination of one sequence from each subset forming a valid game history. Yet traditional IIGs merge these distinct components into a single, monolithic infoset, precluding independent analysis of each dimension, and leaving hand abstraction without a theoretical foundation. To address this gap, we extend the traditional IIG framework with additional constraints.

Before introducing SOOGs, we first define history-related operators on IIGs. For any \(i \in \mathcal{N}_c\):
\begin{itemize}[left=0cm]
    \item \textbf{Trace extraction operator} \(\mathcal{H}_{-i}(\cdot)\) : Retains actions of all players except \(i\) in a history, replacing \(i\)'s actions with wildcard \(\emptyset_i\). Example: For \(h = a_i^1 \cdot a_j^1 \cdot a_k^1 \cdot a_i^2\)\footnote{Omits initial \(h^o\); full form: \(h = h^o \cdot a_i^1 \cdot a_j^1 \cdot a_k^1 \cdot a_i^2\)} (where \(a_i^1,a_i^2\) are \(i\)'s actions, \(a_j^1,a_k^1\) are actions of \(j,k \in \mathcal{N}_c\)),  
  \[
  \mathcal{H}_{-i}(h) = \emptyset_i \cdot a_j^1 \cdot a_k^1 \cdot \emptyset_i.
  \]
    \item \textbf{Sequence extraction operator} \(\tilde{\mathcal{H}}_{-i}(\cdot)\) : Defined as \(\tilde{\mathcal{H}}_{-i}(\cdot) := \mathcal{E} \circ \mathcal{H}_{-i}\), where \(\mathcal{E}\) (wildcard elimination operator) removes all wildcards. For the above \(h\):  
  \[
  \tilde{\mathcal{H}}_{-i}(h) = \mathcal{E}\left(\mathcal{H}_{-i}(h)\right)= \mathcal{E}\left(\emptyset_i \cdot a_j^1 \cdot a_k^1 \cdot \emptyset_i\right) = a_j^1\cdot a_k^1.
  \] 
\end{itemize}

We also define operators focusing on \(i\):  
\begin{itemize}[left=0cm]
    \item \(\mathcal{H}_{i}(\cdot)\) (trace extraction for \(i\)): \(\mathcal{H}_{i}(\cdot) := \mathcal{E} \circ \bigcirc_{\substack{j \in \mathcal{N}_c \\ j \neq i}} \mathcal{H}_{-j}\)\footnote{\(\bigcirc\): sequential composition, \(\bigcirc_{k=1}^N f_k := f_N \circ \dots \circ f_1\)}, retaining \(i\)'s actions with others replaced by wildcards.  
    \item \(\tilde{\mathcal{H}}_{i}(\cdot)\) (sequence extraction for \(i\)): \(\tilde{\mathcal{H}}_{i}(\cdot) := \mathcal{E} \circ \mathcal{H}_{i}\), eliminating wildcards from \(\mathcal{H}_{i}(h)\).  
\end{itemize}

Example for \(h\) above:  
\[
\mathcal{H}_{i}(h) = a_i^1 \cdot \emptyset_j \cdot \emptyset_k \cdot a_i^2, \quad \tilde{\mathcal{H}}_{i}(h) = a_i^1\cdot a_i^2.
\]

We introduce a notion of trace complementarity: Two extracted traces are complementary if, at every position, one contains a concrete action of a specific player $i\in \mathcal{N}_c$ while the other contains the wildcard \(\emptyset_i\) (no overlapping concrete actions or simultaneous wildcards at the same position). For any history \(h \in H\), this implies \(\mathcal{H}_i(h)\) (retaining $i$'s actions) and \(\mathcal{H}_{-i}(h)\) (retaining others' actions) are complementary. Notably, two complementary extracted traces can be spliced into a formally complete history (not guaranteed to be legal), denoted as \(h = \mathcal{H}_i(h) \oplus \mathcal{H}_{-i}(h)\) (where \(\oplus\) resolves positions by prioritizing concrete actions over wildcards).

\begin{definition}[Signal Observation Ordered Games (SOOG)] \label{def:soog}
A signal observation ordered game is an imperfect-information game \(\mathcal{G}\) (satisfying Definition~\ref{def:iig}) with the following additional constraints:

\begin{itemize}[left=0cm]
    \item \(\gamma: H \mapsto \mathbb{N}^+\): A phase partition function. For any \(h \in H\), \(\gamma(h)\) counts the number of chance histories (i.e., \(h' \in H_c\)) along the path from \(h^o\) to \(h\) (including \(h\) if \(h \in H_c\)), defining \(h\)'s phase. Let \(\Gamma = \max_{h \in H} \gamma(h)\) denote the final phase; notably, \(\gamma(h^o) = 1\) (so \(h^o\) is a chance history).
    
    \item Chance actions reveal signals: Let \(\Theta = \{\tilde{\mathcal{H}}_c(h) \mid h \in H\}\) be the total set of signals. Chance actions depend only on revealed signals: for any \(h, h' \in H_c\), if \(\tilde{\mathcal{H}}_c(h) = \tilde{\mathcal{H}}_c(h')\), then \(A(h) = A(h')\) (identical legal actions) and \(\zeta(h, a) =  \zeta(h', a)\) for all \(a \in A(h)\) (identical action probability distributions). Since the chance function \(\zeta\) at chance histories only correlates with already dealt signals, and all legal actions $a$ at chance histories serve to reveal new signals, we can construct a more focused function \(\xi(\theta, \theta')\) to simplify and replace \(\zeta(h, a)\)—where \(\theta \in \Theta\) is the revealed signal at the current chance history, and \(\theta' \in \Theta\) is the new signal revealed by action a.
    
    \item Signal-action separability: For any histories \(h_1, h_2, h_1' \in H\), if \(\tilde{\mathcal{H}}_c(h_1) = \tilde{\mathcal{H}}_c(h_1')\) (i.e., they share identical chance action sequences) and \(\mathcal{H}_{-c}(h_1) = \mathcal{H}_{-c}(h_2)\) (i.e., they share identical non-chance action traces), then there must exist a history \(h_2' \in H\) such that \(\tilde{\mathcal{H}}_c(h_2') = \tilde{\mathcal{H}}_c(h_2)\) and \(\mathcal{H}_{-c}(h_2') = \mathcal{H}_{-c}(h_1')\).
    
    \item Signal observation partitions: Let \(\Psi = (\Psi_i)_{i \in \mathcal{N}}\), where \(\Psi_i\) (signal observation infosets for \(i\)) is a partition of \(\Theta\). Let \(\vartheta = (\vartheta_i)_{i \in \mathcal{N}}\) be a tuple of observation functions, where \(\vartheta_i: \Theta \mapsto \Psi_i\) maps each signal \(\theta \in \Theta\) to its corresponding signal observation infoset \(\psi \in \Psi_i\). Signals in the same \(\psi\in \Psi_i\) are indistinguishable to \(i\).
    
    \item Phase-specific subsets: For each phase \(r \in \{1, \dots, \Gamma\}\), define: 
\(H^{(r)} = \{h \in H \mid \gamma(h) = r\}\) (phase-\(r\) histories), 
\(H_i^{(r)} = H^{(r)} \cap H_i\) (phase-\(r\) histories where \(i\) acts), 
\(Z^{(r)} = Z \cap H^{(r)}\) (phase-\(r\) terminal histories), 
\(\Theta^{(r)} = \{\tilde{\mathcal{H}}_c(h) \mid h \in H^{(r)}\}\) (phase-\(r\) signals), 
\(\Psi_i^{(r)} = \{ \psi \cap \Theta^{(r)} \mid \psi \in \Psi_i \}\) (phase-\(r\) signal observation infosets for \(i\)), 
\(\Psi^{(r)} = (\Psi_i^{(r)})_{i \in \mathcal{N}}\) (phase-\(r\) signal observation partitions), 
\(\mathcal{I}_i^{(r)} = \{I \in \mathcal{I}_i \mid I \subseteq H^{(r)}\}\) (phase-\(r\) infosets for \(i\)), 
and \(\mathcal{I}^{(r)} = (\mathcal{I}_i^{(r)})_{i \in \mathcal{N}}\) (phase-\(r\) information partitions).
    
    \item History indistinguishability criterion: For any \(h, h' \in H_i\) and \(i \in \mathcal{N}\), \(h\) and \(h'\) belong to the same infoset \(I \in \mathcal{I}_i\) if and only if \(\mathcal{H}_{-c}(h) = \mathcal{H}_{-c}(h')\) (identical non-chance action traces) and \(\vartheta_i(\tilde{\mathcal{H}}_c(h)) = \vartheta_i(\tilde{\mathcal{H}}_c(h'))\) (indistinguishable signals).
    
    \item Survival functions: \(\omega = (\omega_i)_{i \in \mathcal{N}}\) is a tuple of survival functions, where for each \(i \in \mathcal{N}\) and \(h \in H\), 
    \[
    \omega_i(h) = \mathbb{I}\{\text{player } i \text{ still participates at } h\},
    \]
    and \(\mathbb{I}\{\cdot\}\) denotes the indicator function (1 if the condition holds, 0 otherwise).
    
    \item Terminal order and utility consistency: In the final phase \(\Gamma\), each signal \(\theta \in \Theta^{(\Gamma)}\) (final signals) induces a total order \(\preceq_\theta\) over \(\mathcal{N}\), which satisfies reflexivity (\(i \preceq_\theta i\) for all \(i\)), totality (for any \(i,j \in \mathcal{N}\), either \(i \preceq_\theta j\) or \(j \preceq_\theta i\)), and transitivity (if \(i \preceq_\theta j\) and \(j \preceq_\theta k\), then \(i \preceq_\theta k\)). For any terminal history \(z \in Z^{(\Gamma)}\) with \(\theta = \tilde{\mathcal{H}}_c(z)\), if \(\omega_i(z)\omega_j(z) = 1\) (both \(i\) and \(j\) survive at \(z\)) and \(i \preceq_\theta j\), then \(u_i(z) \leq u_j(z)\).
\end{itemize}
\end{definition}

Compared to IIGs, SOOGs better capture the core features of hold'em games, with two key refinements aligned to the structure of hold'em games:

First, SOOG's \textbf{signal-action separability} matches "hands vs. betting actions" split in hold'em games. Take a HULH example: Suppose a history $h$ that player 1 holds \(\texttt{A}\heartsuit\texttt{K}\heartsuit\), player 2 holds \(\texttt{J}\heartsuit\texttt{Q}\heartsuit\), and the Flop community cards are \(\texttt{A}\spadesuit\texttt{2}\heartsuit\texttt{3}\heartsuit\); the signal here is \(\theta = \tilde{\mathcal{H}}_c(h) = \texttt{[(A}\heartsuit\texttt{K}\heartsuit\text{), (J}\heartsuit\texttt{Q}\heartsuit\text{)} \mid \texttt{A}\spadesuit\texttt{2}\heartsuit\texttt{3}\heartsuit\text{]}\). SOOG's observation functions map this signal to players' private views: \(\vartheta_1(\theta) = \texttt{[A}\heartsuit\texttt{K}\heartsuit\mid \texttt{A}\spadesuit\texttt{2}\heartsuit\text{3}\heartsuit\text{]}\) (player 1 sees their own hand) and \(\vartheta_2(\theta) = \texttt{[J}\heartsuit\texttt{Q}\heartsuit\mid \texttt{A}\spadesuit\texttt{2}\heartsuit\text{3}\heartsuit\text{]}\) (player 2 sees theirs). Crucially, the betting actions in history \(h\) (e.g., player 1 raises, player 2 calls in Preflop phase) form the non-chance trace \(\mathcal{H}_{-c}(h)\), and separability ensures this trace can combine with other signals (e.g., replacing player 1's hand with \(\texttt{10}\heartsuit\texttt{J}\heartsuit\) to form \(\theta' = \texttt{[(10}\heartsuit\texttt{J}\heartsuit\text{), (J}\heartsuit\texttt{Q}\heartsuit\text{)} \mid \texttt{A}\spadesuit\texttt{2}\heartsuit\text{3}\heartsuit\text{]}\)) to yield a valid hold'em history—unlike IIGs, which treat hands and actions as an indivisible unit.

Second, SOOG's \textbf{total order on final signals} models hold'em's showdown. At showdown, rewards depend on hand strength, which SOOG encodes as a total order \(\preceq_\theta\) over the final signal \(\theta \in \Theta^{(\Gamma)}\) (e.g., player 1's stronger hand gives \(2 \preceq_\theta 1\)). Combined with survival functions (both players survive to showdown, so \(\omega_1(z)=\omega_2(z)=1\)), SOOG's utility constraint \(u_i(z) \leq u_j(z)\) for \(i \preceq_\theta j\) directly reflects the effect of showdown rewards.

These two refinements reveal SOOG signals can be studied independently and have inherent quality hierarchies (even in non-final phases), enabling standalone research and classification as an abstraction technique for this class of imperfect-information games.

\subsection{Signal Observation Abstraction}

We can now model the hand abstraction task using the SOOG framework:

\begin{definition}
In a SOOG \(\mathcal{G}\), \(\alpha = (\alpha_1, \dots, \alpha_N)\) is a signal observation abstraction profile, and \(\Psi^{\alpha}_i\) denotes the set of abstracted signal observation infosets for \(i \in \mathcal{N}\) under \(\alpha\). Each \(\alpha_i: \Theta \mapsto \Psi^{\alpha}_i\) maps a signal \(\theta \in \Theta\) to an abstracted signal observation infoset \(\psi^\alpha \in \Psi^{\alpha}_i\). Furthermore, each \(\psi^\alpha\) can be partitioned into finer signal observation infosets within \(\Psi_i\).
\end{definition}

Given a signal observation abstraction profile \(\alpha = (\alpha_1, \dots, \alpha_n)\) for a SOOG \(\mathcal{G}\), a \textbf{(signal observation) abstracted game} \(\mathcal{G}^\alpha\) is derived by substituting \(\vartheta_i\) with \(\alpha_i\) in \(\mathcal{G}\); signal observation abstraction does not directly alter the game itself but modifies the structure of players' strategy spaces, and since infosets in SOOGs are defined exclusively based on players' signal observations (i.e., the original infoset partition \(\mathcal{I}_i\) is induced by \(\vartheta_i\)), replacing \(\vartheta_i\) with \(\alpha_i\) reshapes the information structure: the abstracted infoset partition \(\mathcal{I}_i^\alpha\) for player \(i\) is defined such that two original infosets \(I, I' \in \mathcal{I}_i\) belong to the same abstracted infoset \(I^\alpha \in \mathcal{I}_i^\alpha\) if and only if there exist histories \(h,h' \in I\) satisfying \(\alpha_i(\tilde{\mathcal{H}}_c(h)) = \alpha_i(\tilde{\mathcal{H}}_c(h'))\) and \(\mathcal{H}_{-c}(h) = \mathcal{H}_{-c}(h')\), and notably, \(\mathcal{I}_i^\alpha\) forms a partition of \(\mathcal{I}_i\)—every original infoset \(I \in \mathcal{I}_i\) is a subset of exactly one abstracted infoset \(I^\alpha \in \mathcal{I}_i^\alpha\) (i.e., \(I \subset I^\alpha\)); in the abstracted game \(\mathcal{G}^\alpha\), each player \(i \in \mathcal{N}\) selects an \(\textbf{abstracted strategy}\) \(\sigma^\alpha_i \in \Sigma^\alpha_i\) (where \(\Sigma^\alpha_i\) denotes the set of all such abstracted strategies for \(i\)), formally, \(\sigma^\alpha_i = \{ \sigma^\alpha_i(I^\alpha, \cdot) \mid I^\alpha \in \mathcal{I}_i^\alpha \}\) with \(\sigma^\alpha_i(I^\alpha, \cdot): A(I^\alpha) \to [0,1]\) being a probability distribution over the legal actions at \(I^\alpha\) (with \(A(I^\alpha) = A(I)\) for all \(I \subset I^\alpha\), as legal actions are invariant across original infosets merged into \(I^\alpha\)) and satisfying \(\sum_{a \in A(I^\alpha)} \sigma^\alpha_i(I^\alpha, a) = 1\) (normalization) for each \(I^\alpha \in \mathcal{I}_i^\alpha\), and when all players select their abstracted strategies, they form an abstracted strategy profile \(\sigma^\alpha = (\sigma^\alpha_1, \ldots, \sigma^\alpha_N)\) with the set of all possible abstracted profiles denoted by \(\Sigma^\alpha = \times_{i \in \mathcal{N}} \Sigma^\alpha_i\); crucially, any abstracted strategy \(\sigma^\alpha_i \in \Sigma^\alpha_i\) can induce a corresponding original strategy \(\sigma_i \in \Sigma_i\) in \(\mathcal{G}\): for every original infoset \(I \in \mathcal{I}_i\) with \(I \subset I^\alpha\) (for some \(I^\alpha \in \mathcal{I}_i^\alpha\)), the action probability distribution of \(\sigma_i\) at \(I\) is set equal to that of \(\sigma^\alpha_i\) at \(I^\alpha\) — i.e., \(\sigma_i(I, a) = \sigma^\alpha_i(I^\alpha, a)\) for all \(a \in A(I)\). A signal observation abstraction profile \(\alpha\) is said to have \(\textbf{perfect-recall}\) if the abstracted game \(\mathcal{G}^\alpha\) induced by \(\alpha\) has perfect-recall; otherwise, it is called an \(\textbf{imperfect-recall}\) abstraction. It should be noted that arbitrary historical information discard represents merely an extreme form of imperfect-recall abstraction.


\section{Evaluation for Hand Abstraction Algorithms} \label{sec:evaluation-poi-froi}

With a suitable mathematical model in place, we can now construct evaluation methods for hand abstraction algorithms. To this end, this section proposes the \textbf{resolution bound} metric, establishes \textbf{potential-aware outcome isomorphism} to estimate the maximum performance limits of existing hand abstraction algorithms, and analyzes their inherent shortcomings.

\subsection{Resolution Bound}

A reasonable evaluation of signal observation abstractions involves identifying the strategy profile with the best performance in the strategy space they construct (e.g., in 2-player zero-sum scenarios, finding the abstraction with the lowest exploitability). Yet this performance-based indirect evaluation is impractical due to computational constraints. We thus need a framework starting directly from the abstractions themselves, with one intuitive direction being to assess their granularity.

\begin{definition}
    In a SOOG, \(\alpha_i\) and \(\beta_i\) are signal observation abstractions for player $i\in \mathcal{N}$. The refinement relationship between \(\alpha_i\) and \(\beta_i\) is defined as follows: If, for any \(\psi^\beta \in \Psi_i^{\beta}\), there exist one or more abstracted signal observation infosets in \(\Psi_i^{\alpha}\) such that their union forms a partition of \(\psi^\beta\), then \(\alpha_i\) is said to refine \(\beta_i\), denoted as \(\alpha_i \sqsupseteq \beta_i\). Furthermore, \(\alpha_i\) is a common refinement of multiple signal observation abstractions \(\alpha_i^{1}, \ldots, \alpha_i^{m}\) for player $i$ if \(\alpha_i \sqsupseteq \alpha_i^{j}\) for all \(j \in \{1, \ldots, m\}\).
\end{definition}

It is hard to judge the quality of two signal observation abstractions without a refinement relationship: in some abstracted signal observation infosets, \(\alpha_i\) is more refined than \(\beta_i\), while in others, \(\beta_i\) is more refined than \(\alpha_i\). However, among those with a refinement relationship, the more refined one has inherent advantages. As \citet{waugh2009abstraction} proved this property for IIGs, we extend it to SOOGs (proof in Appendix~\ref{apdx:proof-of-monotonicity}):

\begin{theorem}[Adapted to SOOG from Theorem 3 in~\cite{waugh2009abstraction}] \label{thm:monotonicity}
Let \(\alpha_i\) and \(\beta_i\) be signal observation abstractions for player $i\in \{1, 2\}$ in a 2-player zero-sum perfect-recall SOOG \(\mathcal{G}\). Let \(\alpha = (\alpha_i, \vartheta_{-i})\) and \(\beta = (\beta_i, \vartheta_{-i})\) be perfect-recall signal observation abstraction profiles, with \(\mathcal{G}^\alpha\) and \(\mathcal{G}^{\beta}\) denoting their induced abstracted games respectively. If \(\alpha_i \sqsupseteq \beta_i\), and \(\sigma^\alpha\) and \(\sigma^\beta\) are the strategies mapped back to the original game from the Nash equilibria of \(\mathcal{G}^\alpha\) and \(\mathcal{G}^{\beta}\) respectively, then \(\epsilon_i(\sigma^\alpha) < \epsilon_i(\sigma^\beta)\).\end{theorem}

Theorem~\ref{thm:monotonicity} imposes several restrictive conditions: it requires player $i$'s opponents adopt no abstraction, and the profiles \(\alpha\), \(\beta\) satisfy perfect-recall. However, in practice, we need not adhere strictly to all these constraints. Fundamentally, Theorem~\ref{thm:monotonicity} gives key intuition: more granular abstractions tend to enable more competitive strategies when solving the abstracted game.

Since signal observation abstractions—i.e., hand abstraction in hold'em games—are all generated automatically by algorithms, the concept of common refinement can be extended to evaluate the performance of these algorithms. For a given SOOG \(\mathcal{G}\), if there exists a signal observation abstraction \(\alpha_i\) such that, for any set of parameters, the signal observation abstraction \(\alpha_i'\) generated by a signal observation abstraction algorithm can be refined by \(\alpha_i\) (i.e., \(\alpha_i \sqsupseteq \alpha_i'\)), then \(\alpha_i\) serves as the common refinement of the algorithm (to be rigorous, this is restricted to the context of \(\mathcal{G}\) and rational player $i$). The common refinement of a signal observation abstraction algorithm defines the upper bound of its ability to distinguish signal observation infosets, which we refer to as the \textbf{resolution bound}. And drawing on Theorem~\ref{thm:monotonicity}, we gain an intuition that the high-quality solution solved from the abstracted game induced by the resolution bound of an observation abstraction algorithm will tend to have better performance than that solved from the abstracted game induced by the abstractions generated by the algorithm itself—though this performance advantage is not strictly guaranteed. Note that a signal observation algorithm with a finer resolution bound does not guarantee better performance; however, if a coarse resolution bound can be found for a signal observation algorithm, it reflects that the algorithm has certain flaws.

\subsection{Potential-Aware Outcome Isomorphism} \label{sec:paoi}

We now construct a flawed signal observation abstraction for a given SOOG \(\mathcal{G}\), and subsequently prove that this flawed abstraction serves as the resolution bound of existing popular signal observation abstraction algorithms that follow the paradigm of arbitrary historical information discard. Through this analysis, we aim to reveal the inherent defects of current algorithms and elaborate on the underlying causes of these limitations.

Without loss of generality, we construct the signal observation abstraction from the perspective of rational player \( i \). We extend the nature's reaching probability notation \(\pi_c(h)\) (defined for histories \(h\)) to signals \(\theta\): let \(\pi_c(\theta) = \prod\nolimits_{\theta' \sqsubseteq \theta} \xi(\theta'', \theta')\), where \(\theta''\) is \(\theta'\)'s immediate predecessor, and \(\sqsubseteq\) denotes signal predecessor (or equal).

First, for any signal observation infoset \(\psi \in \Psi_i^{(r)}\) in the final phase \(r=\Gamma\), we define a winrate outcome feature as
\begin{equation*}
wo_i^{(r)}(\psi) = \big( wo_i^{(r),0}(\psi), wo_i^{(r),1}(\psi), \ldots, wo_i^{(r),N}(\psi) \big), \label{eq:wof}
\end{equation*}
where we denote \(\mathcal{N}_{-i} = \mathcal{N} \setminus \{i\}\). The components are defined as follows:

\begin{itemize}[left=0cm]
    \item \( wo_i^{(r),0}(\psi) \) denotes the probability that player \( i \) ranks lower than at least one other player. Formally:  
      \[
      wo_i^{(r),0}(\psi) = \frac{\sum_{\theta \in \psi} \pi_c(\theta)\cdot\mathbb{I}\!\left\{ \exists j \in \mathcal{N}_{-i} \text{ such that } j \not\preceq_\theta i \text{ and } i \preceq_\theta j \right\}}{\sum_{\theta \in \psi} \pi_c(\theta)}.
      \]  

    \item For \( l > 0 \), \( wo_i^{(r),l}(\psi) \) denotes the probability that player \( i \) ranks no lower than any other player and ranks higher than exactly \( l-1 \) other players. Formally:  
     \begin{align}
      w&o_i^{(r),l}(\psi) = \label{eq:woirl} \\
      &\frac{\sum_{\theta \in \psi} \pi_c(\theta)\cdot\mathbb{I}\!\left\{ \forall j \in \mathcal{N}_{-i}, j \preceq_\theta i \text{ and } \left| \{ j \in \mathcal{N}_{-i} \mid i \preceq_\theta j \} \right| = l-1 \right\}}{\sum_{\theta \in \psi} \pi_c(\theta)}.\nonumber
      \end{align}
\end{itemize}
where \(\left| \{ \cdot \} \right|\) denotes the cardinality of the set (i.e., the count of elements satisfying the condition).

Specifically, in the 2-player scenario, only \(wo_i^{(\Gamma),0}(\psi)\), \(wo_i^{(\Gamma),1}(\psi)\), and \(wo_i^{(\Gamma),2}(\psi)\) remain. In hold'em games, these correspond respectively to the losing rate, tying rate, and winning rate when enumerating the opponent's hands against the player's own hand.

Then, we construct a \textbf{potential-aware outcome feature (PAOF)} for each signal observation infoset in each phase. Signal observation infosets with identical PAOFs are grouped into the same class, while those with distinct PAOFs belong to different classes; the resulting classes are referred to as \textbf{potential-aware outcome isomorphisms (PAOIs)}. The number of distinct PAOIs in phase \(r\) is denoted as \(\mathcal{C}^{(r)}\).

First, for the final phase \(\Gamma\), the PAOF is defined as:  
\[ pf^{(\Gamma)}(\psi) := wo_i^{(\Gamma)}(\psi). \]  

Second, for a non-final phase \(r\) and a signal observation infoset \(\psi \in \Psi_i^{(r)}\), the PAOF \(pf^{(r)}(\psi)\) is a histogram of the PAOIs from the next phase. Specifically, let \(pi_j^{(r+1)}\) denote the set of signal observation infosets contained in the \(j\)-th PAOI of phase \(r+1\). The \(j\)-th component of the PAOF of \(\psi\) is then given by:

\[ pf_j^{(r)}(\psi) = \frac{\sum_{\theta \in \psi} \pi_c(\theta) \cdot \sum_{\psi' \in pi^{(r+1)}_j}\sum_{\theta' \in \psi'} \xi(\theta, \theta')}{\sum_{\theta \in \psi} \pi_c(\theta)\cdot \sum_{\theta' \in \Theta^{(r+1)}} \xi(\theta, \theta')}.\]

\begin{figure}[h]
  \centering
  \includegraphics[width=\linewidth]{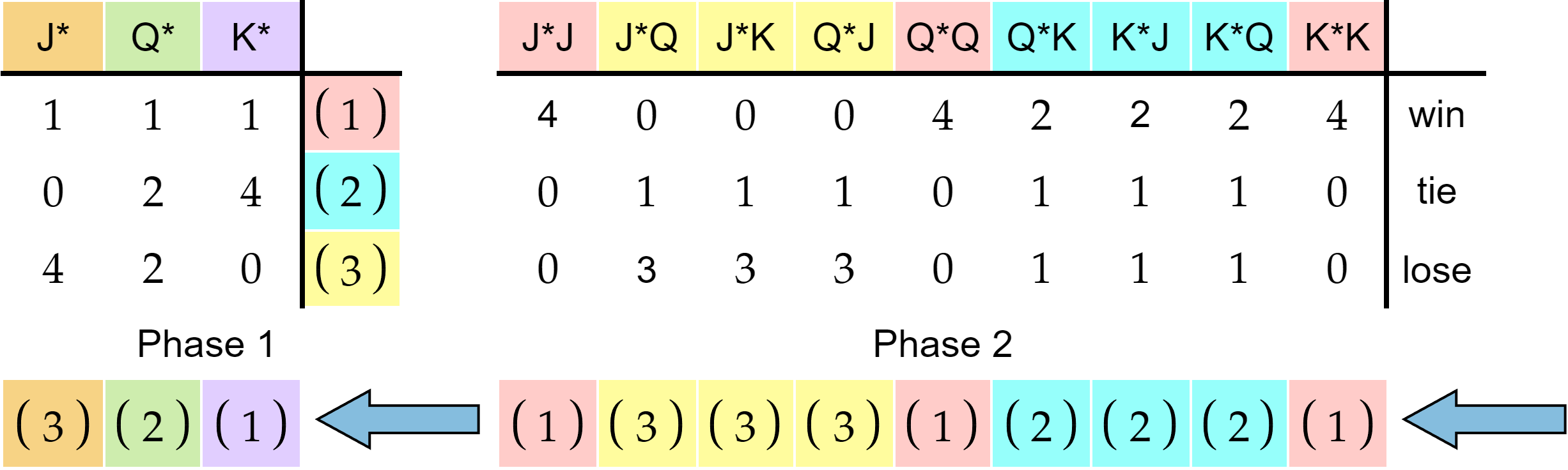}
  \caption{The process of constructing potential-aware outcome isomorphism in Leduc Hold'em.}
  \label{fig:leduc-poi}
\end{figure}

Figure~\ref{fig:leduc-poi} illustrates PAOI construction in Leduc Hold'em (rules in Appendix~\ref{apdx:leduc}). Here, the first letter of each hand denotes a player's private card, and '*' represents the opponent's hidden card. In the second phase, the second letter indicates the revealed community card; each hand corresponds to 4 possible opponent card combinations, yielding 4 signals per infoset. After showdown, we group signal observation infosets into abstracted sets if they share identical win-tie-loss outcome distributions (labeled accordingly). Counts replace win rates here due to equal opponent combinations across hands.

In the first phase, each hand branches into second-phase abstracted signal observation infosets as 5 possible cards are dealt. We group hands with identical branching distributions into abstracted sets; counts again replace probabilities. Leveraging Leduc's perfect-recall (each infoset maps directly to next-phase hands), we enumerate infosets directly instead of individual signals. Notably, this bottom-up construction—where later-phase equivalence classes are built without reference to early-phase information—exemplifies the paradigm of arbitrary historical information discard.

Observations reveal that first-phase abstracted signal observation infosets match the original count (3, indicating strong discriminative power), while the second phase reduces to 3 abstracted sets—far fewer than the original 9—exposing PAOI's limited ability to distinguish infosets. Below are two propositions, with proofs in Appendix~\ref{apdx:proof-of-propositions}.

\begin{proposition} \label{prop:less-ehs}
For a hold'em game modeled as a SOOG $\mathcal{G}$, the PAOI abstraction serves as a resolution bound of algorithm EHS.
\end{proposition}

\begin{proposition} \label{prop:less-paa}
For a hold'em game modeled as a SOOG $\mathcal{G}$, the PAOI abstraction serves as a resolution bound of algorithms PAAA and PAAEMD.
\end{proposition}

Propositions~\ref{prop:less-ehs} and~\ref{prop:less-paa} suggest that the performance of strategies solved under EHS, PAAA and PAAEMD abstractions will likely be weaker than that of their counterparts under PAOI (details of the PAAA, PAAEMD, and EHS algorithms are provided in the Appendix~\ref{apdx:proof-of-propositions}) abstraction. Nevertheless, we next point out and analyze that PAOI itself is not an abstraction with strong performance.

Table~\ref{tbl:defect} presents the number of distinct signal observation infosets across different phases in HULH (ground truth) — with more detailed rules available in Appendix~\ref{apdx:hulh} — and compares it with those abstracted by LI and PAOI. Notably, PAOI's abstraction follow a spindle-shaped distribution: fewer in early/late phases but more in middle phases. By contrast, ground truth infosets exhibit a clear triangular pattern—gradually increasing with game progression—consistent with natural strategic information accumulation. The LI algorithm effectively preserves this pattern, aligning with the game's inherent information growth. PAOI's spindle-shaped distribution, however, deviates sharply from this baseline, implying systematic late-game information loss. This raises concerns about current outcome-based imperfect-recall hand abstraction algorithms (e.g., EHS and PAAEMD) using PAOI as their resolution bound, as such deviation could significantly compromise the effectiveness of their solved strategies—particularly in late phases.

\begin{table}[bt]
\centering
\resizebox{\columnwidth}{!}{
\begin{tabular}{cccc}
\specialrule{1.2pt}{0pt}{0pt}
\hline
Phase   & No Abstraction                  & LI         & PAOI \\ \hline
Preflop & $\binom{52}{2}$=1326              & 169        & 169                                               \\
Flop    & $\binom{52}{2,3}$=25989600$^*$       & 1286792    & 1137132                                                 \\
Turn    & $\binom{52}{2,3,1}$=1195521600    & 55190538   & 2337912                                                \\
River   & $\binom{52}{2,3,1,1}$=5379847200  & 2428287420 & 20687                                                          \\ \hline
\bottomrule[1.2pt]
\end{tabular}
}
\parbox{\columnwidth}{\small * The notation \(\binom{a}{b,c}\) denotes a combinatorial count: given a set of \(a\) elements in total, it first selects \(b\) elements without replacement, then selects \(c\) elements without replacement from the remaining \(a-b\) elements. The total count is equivalent to \(\binom{a}{b}\cdot\binom{a-b}{c}\).
}
\caption{Quantity of signal observation infosets in unabstracted game and abstracted signal observation infosets for various algorithms in HULH.}
\label{tbl:defect}
\vspace{-20pt} 
\end{table}

To understand why PAOI suffers from systematic information loss in late game, we first clarify its core limitation: it arbitrarily discards all historical information, classifying signal observation infosets solely based on the current and future phases. As Figure~\ref{fig:potential_outcome_isomorphism_information} shows, the number of abstracted signal observation infosets under PAOI is jointly determined by two key factors, whose interplay ultimately leads to information loss. On one hand, inherent signal observation infosets grow with game progression, theoretically supporting more distinct classes. On the other hand, the effectiveness of PAOI's classification—i.e., distinguishing whether two signal observation infosets belong to the same abstract class—depends on the volume of data used to compute their PAOFs. For example, in HULH, determining the PAOF of a signal observation infoset (i.e., a hand) requires enumerating all possible rollout scenarios for that infoset: in the River phase, this involves \(\binom{52-7}{2}\) cases (enumerating the opponent's hole cards), while in the Turn phase, it requires \(\binom{52-6}{2,1}\) cases (enumerating both the opponent's hole cards and the final community card in the River phase). As shown in Figure~\ref{fig:potential_outcome_isomorphism_information}, the data volume needed to determine which abstract class a signal (represented by each pink node) belongs to equals the data volume of the subtree rooted at that node—and this volume decreases progressively in later phases. Less data increases the chances of identical PAOFs for distinct signal observation infosets, mistakenly grouping them into the same equivalence class. This not only reduces the number of distinct abstracted signal observation infosets far below the original count but also amplifies \textbf{excessive abstraction} in late game, ultimately resulting in the spindle-shaped distribution of PAOI classes observed earlier.

\begin{figure}[bt]
\centering
\includegraphics[width=0.5\textwidth]{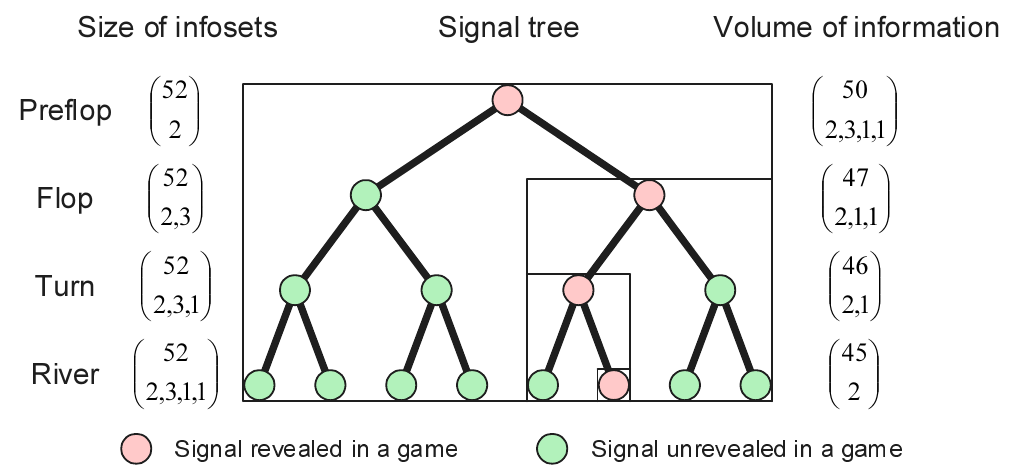}
\caption{Two factors influencing the quantity of distinct abstracted signal observation infosets for algorithms that only use current and future information.}
\label{fig:potential_outcome_isomorphism_information}
\end{figure}

\subsection{Full-Recall Outcome Isomorphism}

We have identified the underlying reasons for PAOI's excessive abstraction phenomenon. To address this, we improve PAOI by incorporating historical information and construct a new abstraction: \textbf{full-recall outcome isomorphism (FROI)}, a special case of \textbf{k-recall outcome isomorphism (k-ROI)}.

k-ROI is developed from PAOI. For any \(\psi \in \Psi_i^{(r)}\) of phase \(r\) in a SOOG \(\mathcal{G}\) (for player \(i\)), we construct its \textbf{\(k\)-recall outcome feature (k-ROF)} — a \(k\)-dimensional vector — defined as:  
\begin{equation}
rf_i^{(r,k)}(\psi) = \big( pi^{(r)}(\psi^{(r)}), pi^{(r-1)}(\psi^{(r-1)}), \ldots, pi^{(r-k)}(\psi^{(r-k)}) \big), \nonumber
\end{equation}  
where \(\psi^{(r')}\) (for \(r' < r\)) denotes the predecessor of \(\psi\) at phase \(r'\) (by the perfect-recall assumption, the predecessors of all signals in \(\psi\) at the same phase \(r'\) belong to the same signal observation infoset), and \(pi^{(r')}(\psi^{(r')})\) represents the label (index) of the PAOI class to which \(\psi^{(r')}\) belongs.

\begin{figure}[tb]
\centering
\includegraphics[width=\linewidth]{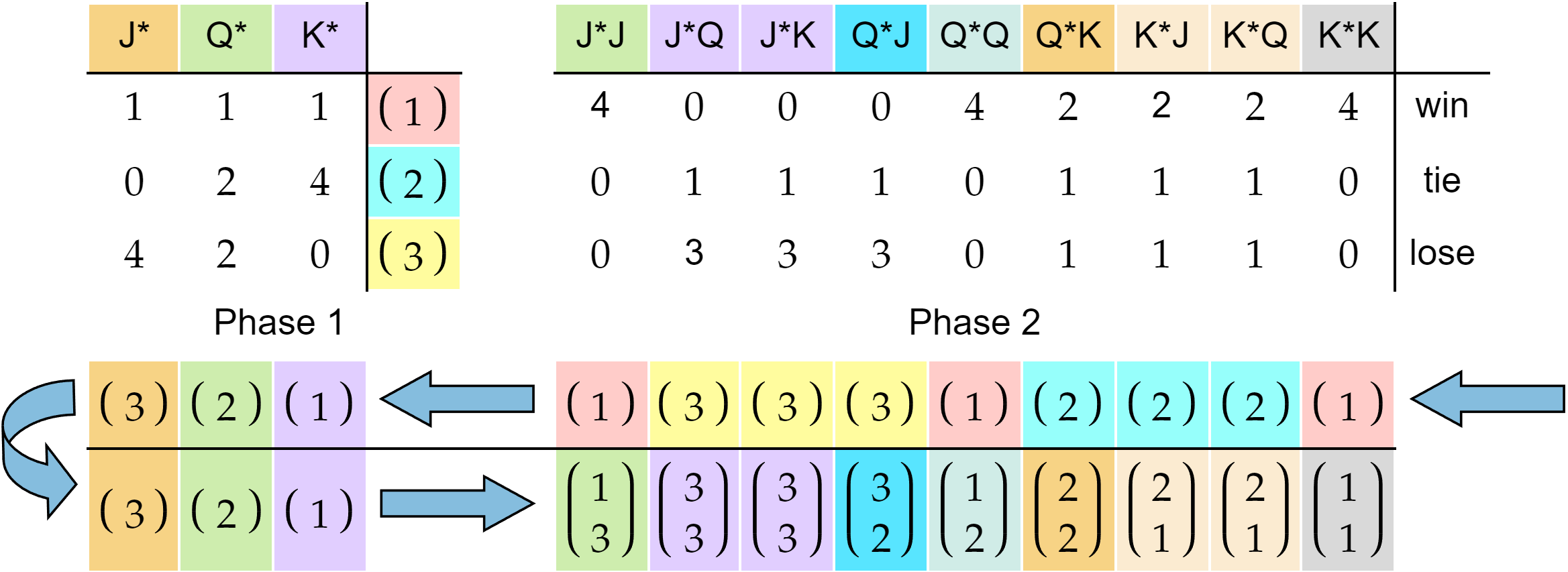}
\caption{The process of constructing full-recall outcome isomorphism in Leduc Hold'em.}
\label{fig:leduc_kroi}
\vspace{-15pt} %
\end{figure}

Figure~\ref{fig:leduc_kroi} illustrates the process of constructing 0-ROI for hands in phase 1 and 1-ROI for hands in phase 2 of Leduc Hold'em. It can be observed that this is a two-step process: first, we compute the PAOFs and PAOIs for the current phase and all preceding phases in a bottom-up manner; then, we calculate the k-ROF and k-ROI for each hand in a top-down manner. For phase $r$, we refer to the \((r-1)\)-ROI as \textbf{full-recall outcome isomorphism (FROI)}. Notably, "full-recall" here does not imply perfect-recall in the abstracted game; it merely indicates inclusion of all prior phase information. This distinction is critical: future clustering-based abstraction algorithms targeting FROI as their resolution bound will likely remain imperfect-recall, yet they demonstrate a viable path for integrating historical context into imperfect-recall abstraction frameworks.

Table~\ref{tbl:krwil} presents the number of distinct abstracted classes identifiable by k-ROI (with \(k = 0, \dots, r-1\)) in phase r of HULH. It is evident that as the value of k increases (i.e., more historical information is incorporated), the number of distinct abstracted classes identifiable by k-ROI grows. Moreover, the number of abstracted classes of FROI across different phases—169, 1241210, 42040233, and 638585633—exhibits the triangular pattern observed in both the no abstraction and LI (Table~\ref{tbl:defect}).

\begin{table*}[ht]
\centering
\resizebox{\textwidth}{!}{%
\begin{tabular}{ccccccccccc}
\specialrule{1.2pt}{0pt}{0pt}
\hline
         & Preflop & \multicolumn{2}{c}{Flop} & \multicolumn{3}{c}{Turn}     & \multicolumn{4}{c}{River}                \\
Phase    & 1 & \multicolumn{2}{c}{2} & \multicolumn{3}{c}{3}     & \multicolumn{4}{c}{4} \\
         \cmidrule(lr){2-2} \cmidrule(lr){3-4} \cmidrule(lr){5-7} \cmidrule(lr){8-11}
Recall (k)   & 0       & 0           & 1           & 0       & 1        & 2        & 0     & 1        & 2         & 3         \\ 
k-ROI     & 169     & 1137132     & 1241210     & 2337912 & 38938975 & 42040233 & 20687 & 39792212 & 586622784 & 638585633 \\ \hline \bottomrule[1.2pt]
\end{tabular}%
}
\caption{The number of signal observation infoset equivalence classes identified by k-ROI in each phase and $k$ of HULH.}
\label{tbl:krwil}
\vspace{-10pt} %
\end{table*}

\section{Experiment}\label{sec:experiment}

To provide more direct evidence that abstraction methods arbitrarily discarding historical information suffer significant performance deficiencies, we evaluate the performance of \(\epsilon\)-Nash equilibria solved in abstracted games when applied to the original game.

First, we select the experimental environment. Given HULH's excessive complexity, we opt for a simplified setting. Notably, conventional simplified environments (e.g., Kuhn Poker, Leduc Hold'em, Flop Hold'em) have no more than 2 phases. As shown in Table~\ref{tbl:krwil}, even in HULH (a game with an extensive range of hands), excessive abstraction is not prominent at the 2-phase stage (1137132 vs. 1241210). We thus construct a 3-phase variant named Numeral211 Hold'em, whose rules are in Appendix~\ref{apdx:numeral211}. Table~\ref{tbl:krwi-kroi-numeral211} presents the number of (abstracted) signal observation infosets identified by various methods (and the unabstracted setting) in Numeral211 Hold'em, revealing a significant discrepancy between PAOI (100, 2250, 3957) and FROI (100, 2260, 51228)—especially in phase 3.

\begin{table}[bt]
  \centering
  \resizebox{\columnwidth}{!}{%
\begin{tabular}{ccccccc}
\specialrule{1.2pt}{0pt}{0pt}
\hline
         & Phase 1 & \multicolumn{2}{c}{Phase 2} & \multicolumn{3}{c}{Phase 3}                   \\
No Abstraction & 780     & \multicolumn{2}{c}{29640} & \multicolumn{3}{c}{1096680}       \\
LI       & 100     & \multicolumn{2}{c}{2260} & \multicolumn{3}{c}{62020}       \\
         \cmidrule(lr){2-2} \cmidrule(lr){3-4} \cmidrule(lr){5-7} 
Recall   & 0       & 0           & 1           & 0       & 1        & 2               \\ 
k-ROI     & 100     & 2250        & 2260     & 3957 & 51176 & 51228  \\ \hline   \bottomrule[1.2pt]
\end{tabular}%
}
\caption{Quantity of signal observation infosets in unabstracted game and abstracted signal observation infosets for various algorithms in Numeral211 Hold'em.}
\label{tbl:krwi-kroi-numeral211}
\vspace{-10pt} %
\end{table}

We conduct two evaluations for each set of signal observation abstractions to be assessed, denoted as \(\alpha = (\alpha_1, \alpha_2)\). The first is \textbf{asymmetric abstraction}, which follows the setup in Theorem~\ref{thm:monotonicity}: we construct two abstractions \(\alpha' = (\alpha_1, \vartheta_2)\) and \(\alpha'' = (\vartheta_1, \alpha_2)\) (where \(\vartheta_1/\vartheta_2\) denote no abstraction for player 1/2), use CFR to solve $\epsilon$-Nash equilibria \(\sigma'\) and \(\sigma''\) for abstracted games \(\mathcal{G}^{\alpha'}\) and \(\mathcal{G}^{\alpha''}\), then concatenate these into a joint strategy \(\sigma = (\alpha'_1, \alpha''_2)\) and evaluate its exploitability in the original game \(\mathcal{G}\). This method has theoretical guarantees—finer abstractions yield stronger-performing strategies—but it does not reflect typical hand abstraction scenarios: in large-scale games, solving strategies with no opponent abstraction is computationally infeasible due to enormous space overhead. We thus adopt a second scenario, \textbf{symmetric abstraction}: given \(\alpha\), we solve the $\epsilon$-Nash equilibrium of the abstracted game \(\mathcal{G}^{\alpha}\) and evaluate its exploitability in \(\mathcal{G}\). This scenario is susceptible to \textbf{abstraction pathology}~\cite{waugh2009abstraction}—a rare phenomenon where finer abstractions may unexpectedly produce worse-performing strategies.

\begin{figure}[h]
  \centering
  \includegraphics[width=\linewidth]{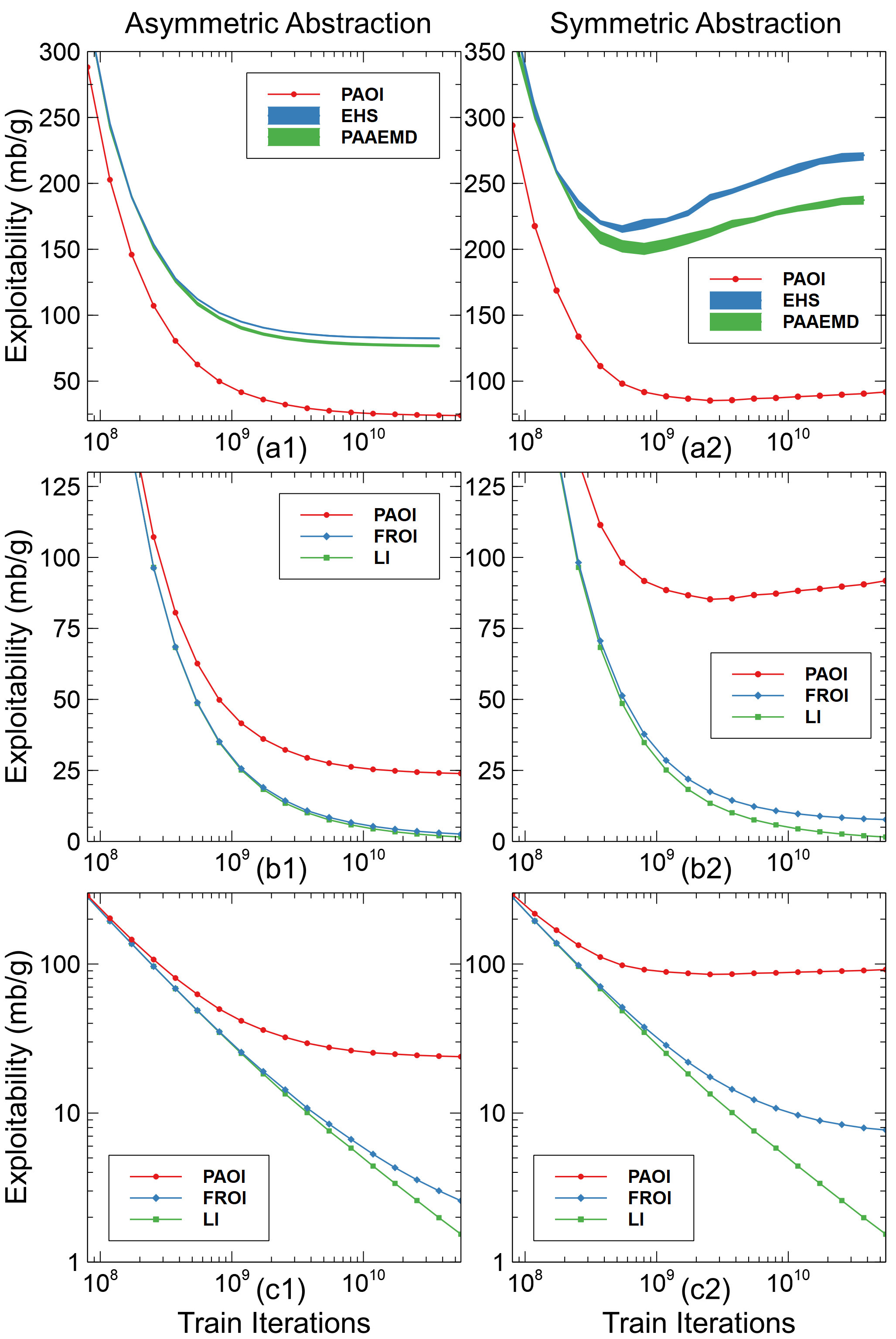}
  \caption{Experimental results on Numeral211 Hold'em.}
  \label{fig:experiment}
  \vspace{-20pt} %
\end{figure}

In the first experiment, we verify that the performance of mainstream signal observation abstraction algorithms is hardly superior to that of PAOI. Specifically, we use the EHS and PAAEMD algorithms: we apply no abstraction in phase 1, and construct abstraction classes accounting for 1/10 of PAOI's maximum discriminative capacity in phases 2 and 3 (i.e., 100-225-396). Given the sensitivity of clustering algorithms to initial values, we generate 5 independent groups of abstractions for each algorithm. Subfigures (a1) and (b1) of Figure~\ref{fig:experiment} present the results under the asymmetric abstraction and symmetric abstraction scenarios, respectively; it can be observed that the performance of both EHS and PAAEMD is significantly weaker than that of PAOI—particularly in terms of exploitability, where higher values indicate poorer performance.  


In the second experiment, we examine performance among PAOI, FROI, and LI. Notably, LI is a lossless abstraction based on hand suit rotation and serves as the ground truth—crucially, it cannot be further compressed via clustering. Figure~\ref{fig:experiment}(a2, b2) show results for asymmetric and symmetric abstraction: PAOI-derived solutions perform significantly worse than FROI and LI, while FROI's performance is very close to LI's—especially in asymmetric abstraction. To distinguish their subtle gap, Figure~\ref{fig:experiment}(a3, b3) present results on a logarithmic scale for clearer resolution of minor differences.

These results demonstrate that discarding historical information systematically degrades the performance of abstraction algorithms, whereas incorporating historical information elevates the performance upper bound of such algorithms.

\section{Conclusion}

This paper presents a suitable mathematical model and evaluation metric for the hand abstraction task, and further identifies a low-resolution bound in current outcome-based imperfect-recall algorithms—PAOI—which limits AI performance in signal observation abstraction, especially when scaling computational resources. This limitation stems from the arbitrary discarding of historical information. To address this, we introduce FROI, a hand abstraction that incorporates historical information and achieves a finer resolution bound—offering key insights into how to weave historical context into abstraction design. Experiments confirm that such oversimplified bounds (e.g., PAOI) impair AI performance, validating the value of FROI's approach to historical information.

Notably, FROI is a hand abstraction rather than a hand abstraction algorithm, meaning it cannot adjust the number of equivalence classes (e.g., via clustering) as algorithmic methods do—this points to a critical direction for future work. We will develop signal abstraction algorithms capable of exceeding the PAOI bound and targeting the FROI bound.

\bibliographystyle{ACM-Reference-Format} 
\bibliography{sample}


\newpage
\onecolumn
\appendix

\section{Rules of Hold'em Games}\label{apdx:game_rules}

\subsection{Heads-Up Limit Hold'em (HULH)} \label{apdx:hulh}
Heads-up limit hold'em (HULH) is a two-player, fixed-bet variant of Texas Hold'em, preserving the core mechanics of community cards and standard poker hand rankings while optimizing for two-player strategic interactions. Its complete rule set is outlined below:
\begin{enumerate}
\item \textbf{Betting Phases:} The game unfolds across four sequential phases—Preflop, Flop, Turn, and River. Each phase commences after the corresponding cards are dealt, with players taking turns to act (fold, check, call, bet, or raise).
\item \textbf{Deck and Hand Construction:}
A standard 52-card deck (excluding Jokers) is used. Before the first betting phase (Preflop), each player receives two "private cards" that only they can view.
Over the course of the game, five shared "community cards" are revealed in phases: three cards (called the "Flop") are exposed first, followed by a single "Turn" card, and finally a single "River" card. To determine the winner, each player combines their two hole cards with the five community cards to form the strongest possible 5-card hand.
\item \textbf{Hand Rankings:} Follow the standard poker hierarchy, ranging from Straight Flush (the highest-ranked hand) to High Card (the lowest-ranked hand). Detailed rankings and their descriptions are provided in Table \ref{tbl:hulh_hand_rank}.
\item \textbf{Blind Structure:} To initiate the pot, two forced bets (blinds) are required: the Small Blind (SB) and Big Blind (BB). The SB is set to half the value of the BB, and these bets are posted by the two players before any cards are dealt.
\item \textbf{Action Order:} Prior to the Flop (Preflop), the player who posted the SB acts first. For all subsequent phases (Flop, Turn, River), the action starts with the player who posted the BB, alternating thereafter.
\item \textbf{Standard AI Research Setup:} In academic studies, HULH typically uses a 100-chip SB, with each player beginning the game with a 20,000-chip stack—equivalent to 200 times the BB value (200 BB effective stack depth).
\item \textbf{Fixed Betting Rules:}
Preflop: All bets and raises are capped at the BB amount (e.g., if the BB is 200 chips, Preflop bets/raises are limited to 200 chips).
Postflop (Flop, Turn, River): The size of bets and raises doubles to twice the BB value, maintaining consistent incremental increases across phases.
Raise Limits: Each betting phase allows a maximum of one initial bet plus three raises (e.g., bet → raise → re-raise → cap raise). This constraint prevents unbounded betting sequences and ensures strategic focus on hand strength rather than bet sizing.
\end{enumerate}

\begin{table}[tb]
\centering
\begin{tabular}{cccp{5cm}c}
\hline
\textbf{Rank} & \textbf{Hand}         & \textbf{Prob.} & \multicolumn{1}{c}{\textbf{Description}}                                                                                     & \textbf{Example}                                      \\ \hline
1             & Straight Flush       & 0.0015\%       & Five consecutive-rank cards of the same suit. Ties are resolved by comparing the highest card in the sequence.                                                                 & \(\texttt{A}\spadesuit\texttt{K}\spadesuit\texttt{Q}\spadesuit\texttt{J}\spadesuit\texttt{T}\spadesuit\) \\
2             & Four of a Kind       & 0.0240\%       & Four cards of identical rank. Ties are broken by the rank of the four matching cards.                                                                                              & \(\texttt{8}\spadesuit\texttt{8}\heartsuit\texttt{8}\diamondsuit\texttt{8}\clubsuit\texttt{K}\heartsuit\) \\
3             & Full House           & 0.1441\%       & A combination of three cards of one rank and two cards of another (a "three-of-a-kind" plus a "pair"). Ties are resolved first by the rank of the three-of-a-kind, then by the rank of the pair.                                                                          & \(\texttt{Q}\spadesuit\texttt{Q}\heartsuit\texttt{Q}\diamondsuit\texttt{5}\clubsuit\texttt{5}\spadesuit\) \\
4             & Flush                & 0.1965\%       & Five cards of the same suit, with no consecutive rank sequence. Ties are broken by comparing the highest card, then the next highest, and so on.                                                     & \(\texttt{9}\heartsuit\texttt{7}\heartsuit\texttt{6}\heartsuit\texttt{4}\heartsuit\texttt{2}\heartsuit\) \\
5             & Straight             & 0.3925\%       & Five consecutive-rank cards of mixed suits. Ties are resolved by the highest card in the sequence.                                                                      & \(\texttt{K}\spadesuit\texttt{Q}\heartsuit\texttt{J}\clubsuit\texttt{T}\diamondsuit\texttt{9}\heartsuit\) \\
6             & Three of a Kind      & 2.1128\%       & Three cards of identical rank. Ties are broken first by the rank of the three matching cards, then by the highest remaining card ("kicker"), and finally by the second remaining kicker.                                      & \(\texttt{7}\spadesuit\texttt{7}\heartsuit\texttt{7}\diamondsuit\texttt{K}\clubsuit\texttt{4}\heartsuit\) \\
7             & Two Pair             & 4.7539\%       & Two separate pairs of identical-rank cards. Ties are resolved first by the rank of the higher pair, then by the rank of the lower pair, and finally by the remaining kicker.                                                              & \(\texttt{J}\spadesuit\texttt{J}\heartsuit\texttt{5}\clubsuit\texttt{5}\spadesuit\texttt{2}\heartsuit\) \\
8             & One Pair             & 42.2569\%      & Two cards of identical rank. Ties are broken first by the rank of the pair, then by the highest kicker, followed by the next highest, and so on.                                          & \(\texttt{9}\spadesuit\texttt{9}\heartsuit\texttt{A}\clubsuit\texttt{K}\spadesuit\texttt{3}\heartsuit\) \\
9             & High Card            & 50.1177\%      & No pairs or higher-ranked combinations. Ties are resolved by comparing the highest card, then the next highest, and so on.                                                                                & \(\texttt{A}\spadesuit\texttt{K}\heartsuit\texttt{J}\clubsuit\texttt{8}\diamondsuit\texttt{3}\heartsuit\) \\ \hline
\end{tabular}
\caption{Hand ranks of HULH.}
\label{tbl:hulh_hand_rank}
\end{table}

\subsection{Leduc Hold'em} \label{apdx:leduc}
Leduc Hold'em is a simplified two-player poker variant, widely used in imperfect-information game research due to its balance of strategic depth and computational tractability (it avoids the complexity of full Texas Hold'em while retaining core abstraction challenges). Its rules are as follows:
\begin{enumerate}

\item \textbf{Deck Composition:} Uses a stripped-down 6-card deck, consisting of two suits (e.g., hearts and spades) and three ranks: Jack (J), Queen (Q), and King (K). Each rank-suit combination (e.g., \(\texttt{J}\heartsuit\), \(\texttt{J}\spadesuit\), \(\texttt{Q}\heartsuit\), \(\texttt{Q}\spadesuit\), \(\texttt{K}\heartsuit\), \(\texttt{K}\spadesuit\)) is present exactly once, resulting in 3 ranks × 2 suits = 6 unique cards.
\item \textbf{Betting Phases:} The game features two distinct betting phases—"Preflop" (before community cards are revealed) and a "Postflop" (after a single community card is exposed).
\item \textbf{Card Dealing:}\\
Preflop: Each player receives one private card, dealt face down.\\
Postflop: A single "community card" is dealt face up in the center of the table, shared by both players.
\item \textbf{Hand Rankings:} Hand strength is determined by combining the player's hole card with the community board card, using the following hierarchy (from highest to lowest):\\
\textbf{Pair:} Two cards of the same rank (e.g., hole card \(\texttt{J}\heartsuit\)+ board card \(\texttt{J}\spadesuit\)).\\
\textbf{High Card:} No pair—hand strength is determined by the rank of the hole card (since the board card is shared, ties are resolved by the hole card rank: K > Q > J).
\item \textbf{Blind and Ante Structure:} The game uses a single forced bet (ante) instead of blinds: each player posts 1 chip into the pot at the start of the hand to initiate action.
\item \textbf{Betting Rules:}
Both betting phases (Preflop and Postflop) use a fixed bet size, typically set to 2 chips (double the ante).
Each phase allows a maximum of one bet and one raise (e.g., bet → raise → cap), with no additional raises permitted. Players may fold, call the current bet, or raise (if no raise has been made in the phase).
\item \textbf{Showdown:} If neither player folds by the end of the Postflop phase, both players reveal their hole cards. The player with the higher-ranked hand wins the entire pot; if hands are tied (e.g., both have high-card Q), the pot is split evenly between the two players.
\end{enumerate}

\subsection{Numeral211 Hold'em} \label{apdx:numeral211}
Numeral211 Hold'em is a custom poker variant designed for hand abstraction research. It is more complex than simple games like Kuhn Poker but substantially less computationally intensive than HULH, making it an ideal testbed for evaluating abstraction algorithms. Its rules are detailed below:
\begin{enumerate}
\item \textbf{Ante Structure:} To start each hand, both players post a mandatory 5-chip ante into the pot, ensuring there is always a prize to compete for.
\item \textbf{Betting Phases:} The game progresses through three sequential betting phases—Preflop (phase 1), Flop (phase 2), and Turn (phase 3). Each phase begins after the corresponding cards are dealt, with players acting in turn (fold, check, call, bet, or raise).
\item \textbf{Deck and Hand Construction:}
Derived from a modified standard 52-card deck: all Kings, Queens, and Jacks are removed, leaving 40 cards total. The remaining cards include ranks 2 through 9, 10 (\texttt{T}), and Ace (\texttt{A}), across four suits.\\
Preflop: Each player receives one private card.\\
Community Cards: Two shared community cards are revealed in phases: one card during the Flop phase and one card during the Turn phase. Players combine their single hole card with the two community cards to form a 3-card hand, with strength determining the winner.
\item \textbf{Hand Rankings:} Follow a simplified 3-card hierarchy (from highest to lowest), as detailed in Table \ref{tbl:numeral211_hand_rank}: Straight Flush, Three of a Kind, Straight, Flush, Pair, and High Card.
\item \textbf{Betting Options:} Players may fold, check, call, bet, or raise in all three phases (similar to HULH). Key constraints include:
Fixed bet sizes: 10 chips for Preflop bets/raises, and 20 chips for Postflop (Flop and Turn) bets/raises.
Raise Limits: Each betting phase allows a maximum of four total bets/raises (e.g., bet → raise → re-raise → cap raise), preventing excessive betting sequences.
\item \textbf{Showdown:} If both players choose to call the final bet of the Turn phase, they reveal their hole cards. The player with the higher-ranked 3-card hand (private cards + community cards) wins the entire pot; in the case of a tie, the pot is split equally between the two players.
\end{enumerate}
\begin{table}[tb]
\centering
\begin{tabular}{cccp{5cm}c}
\hline
\textbf{Rank} & \textbf{Hand} & \textbf{Prob.} & \multicolumn{1}{c}{\textbf{Description}} & \textbf{Example} \\ \hline
1 & Straight Flush & 0.321\% & Three consecutive-rank cards of the same suit. Ties are resolved by the highest card in the sequence. & \(\texttt{T}\spadesuit\texttt{9}\spadesuit\texttt{8}\spadesuit\) \\
2 & Three of a Kind & 1.587\% & Three cards of identical rank. Ties are broken by the rank of the matching cards. & \(\texttt{T}\spadesuit\texttt{T}\heartsuit\texttt{T}\clubsuit\) \\
3 & Straight & 4.347\% & Three consecutive-rank cards of mixed suits. Ties are resolved by the highest card rank. & \(\texttt{T}\spadesuit\texttt{9}\heartsuit\texttt{8}\clubsuit\) \\
4 & Flush & 15.799\% & Three cards of the same suit, with no consecutive rank sequence. Ties are broken by the highest card, then the second highest, and finally the third highest. & \(\texttt{T}\spadesuit\texttt{8}\spadesuit\texttt{6}\spadesuit\) \\
5 & Pair & 34.065\% & Two cards of identical rank. Ties are resolved first by the rank of the pair, then by the rank of the remaining third card. & \(\texttt{T}\spadesuit\texttt{T}\heartsuit\texttt{8}\clubsuit\) \\
6 & High Card & 43.881\% & No pairs or higher-ranked combinations. Ties are broken by comparing the highest card, then the second highest, and finally the third highest. & \(\texttt{T}\spadesuit\texttt{8}\heartsuit\texttt{6}\clubsuit\) \\ \hline
\end{tabular}
\caption{Hand ranks of Numeral211 Hold'em.}
\label{tbl:numeral211_hand_rank}
\end{table}

\section{Proof of Theorem~\ref{thm:monotonicity}} \label{apdx:proof-of-monotonicity}

Theorem~\ref{thm:monotonicity} states: Let \(\alpha_i\) and \(\beta_i\) be signal observation abstractions for player $i\in \{1, 2\}$ in a 2-player zero-sum perfect-recall SOOG \(\mathcal{G}\). Let \(\alpha = (\alpha_i, \vartheta_{-i})\) and \(\beta = (\beta_i, \vartheta_{-i})\) be perfect-recall signal observation abstraction profiles, with \(\mathcal{G}^\alpha\) and \(\mathcal{G}^{\beta}\) denoting their induced abstracted games respectively. If \(\alpha_i \sqsupseteq \beta_i\), and \(\sigma^\alpha\) and \(\sigma^\beta\) are the strategies mapped back to the original game from the Nash equilibria of \(\mathcal{G}^\alpha\) and \(\mathcal{G}^{\beta}\) respectively, then \(\epsilon_i(\sigma^\alpha) < \epsilon_i(\sigma^\beta)\).

\begin{proof}
To prove Theorem \ref{thm:monotonicity}, we first link abstraction refinement to strategy space inclusion, then use exploitability definition, minimax properties for perfect-recall zero-sum SOOGs, and this inclusion to derive the result. By premise \(\alpha_i \sqsupseteq \beta_i\) (\(\alpha_i\) is finer than \(\beta_i\)), finer infoset partitioning implies all \(\beta_i\)-consistent strategies are \(\alpha_i\)-consistent, while \(\alpha_i\) allows additional granular strategies—thus player \(i\)'s strategy space satisfies \(\Sigma_i^\alpha \supseteq \Sigma_i^\beta\). Since \(\alpha = (\alpha_i, \vartheta_{-i})\) and \(\beta = (\beta_i, \vartheta_{-i})\) share \(\vartheta_{-i}\), the opponent's strategy space is fixed (\(\Sigma_{-i}^\alpha = \Sigma_{-i}^\beta = \Sigma_{-i}\)), so we focus on player \(i\)'s abstraction impact.
We begin with the standard exploitability definition for player \(i\) in \(\mathcal{G}\), and derive the key chain:
\[
\begin{aligned}
\epsilon_i(\sigma^\alpha) &= v_i^* - \min_{\sigma_{-i} \in \Sigma_{-i}} u_i(\sigma^\alpha, \sigma_{-i}), \quad \epsilon_i(\sigma^\beta) = v_i^* - \min_{\sigma_{-i} \in \Sigma_{-i}} u_i(\sigma^\beta, \sigma_{-i}) \\
\implies v_i^* - \epsilon_i(\sigma^\alpha) &= \min_{\sigma_{-i} \in \Sigma_{-i}} u_i(\sigma^\alpha, \sigma_{-i}) \\
&= \max_{\sigma_i' \in \Sigma_i^\alpha} \min_{\sigma_{-i} \in \Sigma_{-i}} u_i(\sigma_i', \sigma_{-i}) \quad \text{(Minimax property $\sigma^\alpha$ is $\mathcal{G}^\alpha$ equilibrium, maximizing min utility over $\Sigma_i^\alpha$)} \\
&\ge \max_{\sigma_i' \in \Sigma_i^\beta} \min_{\sigma_{-i} \in \Sigma_{-i}} u_i(\sigma_i', \sigma_{-i}) \\
&= \min_{\sigma_{-i} \in \Sigma_{-i}} u_i(\sigma^\beta, \sigma_{-i}) \quad \text{(Minimax property $\sigma^\beta$ is $\mathcal{G}^\beta$ equilibrium, matching max min utility over $\Sigma_i^\beta$)} \\
&= v_i^* - \epsilon_i(\sigma^\beta) \\
\implies \epsilon_i(\sigma^\alpha) \le \epsilon_i(\sigma^\beta)
\end{aligned}
\]

\end{proof}

\section{Proofs of Propositions}
\label{apdx:proof-of-propositions}
In this section, we prove Proposition~\ref{prop:less-ehs} (PAOI is EHS's resolution bound) and Proposition~\ref{prop:less-paa} (PAOI is PAAA/PAAEMD's resolution bound), using the SOOG framework and main-text notation. To analyze signal extension across phases (key for proving PAOI's role), we first define extended signals: for any phase \(r' \geq r\), let
\begin{equation}
S_\psi^{(r')} := \left\{ \theta \in \Theta^{(r')} \mid \exists \theta' \in \psi,\ \theta' \sqsubseteq \theta \right\}, \label{eq:extended-signals}
\end{equation}
where \(\psi \in \Psi_i^{(r)}\) (phase-\(r\) signal observation infoset for some $i$), \(\theta' \sqsubseteq \theta\) means \(\theta'\) is a prefix of \(\theta\) (SOOG signal sequence), and \(S_\psi^{(r')}\) denotes all phase-\(r'\) signals extending any in \(\psi\). This notation is used in subsequent proofs to formalize current-future signal relationships for abstraction consistency.

\subsection{Proof of Proposition~\ref{prop:less-ehs}}
\label{subsec:proof-less-ehs}

Proposition~\ref{prop:less-ehs} states: For a hold'em game modeled as a SOOG $\mathcal{G}$, the PAOI abstraction serves as a resolution bound of algorithm EHS. We first clarify the core logic of the EHS algorithm, then proceed with the proof.

\paragraph{EHS Algorithm}
EHS (Expected  Hand Strength) is a foundational lossy hand abstraction algorithm for 2-player hold'em games. It quantifies hand strength using \textbf{expected showdown equity} and clusters signals into contiguous equity ranges. Its key steps (adapted to the SOOG framework) are as follows:

\begin{enumerate}[left=0cm]
\item \textbf{Expected Showdown Equity Calculation}:  
   For any signal observation infoset \(\psi \in \Psi_i^{(r)}\) (player \(i\)'s phase-\(r\) signal set), EHS computes the expected showdown equity as:
   \[
   \text{equity}(\psi) = w(\psi) + \frac{1}{2}t(\psi) + 0 \cdot l(\psi).
   \]
   where:
   \begin{itemize}
       \item \(w(\psi)\) (winning rate): Probability that player \(i\) wins at showdown if the game proceeds from \(\psi\);
       \item \(t(\psi)\) (tying rate): Probability that player \(i\) ties at showdown if the game proceeds from \(\psi\);
       \item \(l(\psi)\) (losing rate): Probability that player \(i\) loses at showdown if the game proceeds from \(\psi\).  
   \end{itemize}
   
    We further formalize \(w(\psi)\), \(t(\psi)\), and \(l(\psi)\). For 2-player games (where \(-i\) denotes the opponent of \(i\), with \(-i = 2\) if \(i = 1\) and \(-i = 1\) if \(i = 2\)), the three rates are explicitly defined using final phase (\(r' = \Gamma\)) extended signals, weighted by nature's contribution \(\pi_c(\theta)\):
   \[
   w(\psi) = \frac{\sum_{\theta \in S_\psi^{(\Gamma)}} \pi_c(\theta)\cdot\mathbb{I}\!\left\{ i \not\preceq_\theta -i \text{ and } -i \preceq_\theta i \right\}}{\sum_{\theta \in S_\psi^{(\Gamma)}} \pi_c(\theta)},
   \]
   \[
   t(\psi) = \frac{\sum_{\theta \in S_\psi^{(\Gamma)}} \pi_c(\theta)\cdot\mathbb{I}\!\left\{ i \preceq_\theta -i \text{ and } -i \preceq_\theta i \right\}}{\sum_{\theta \in S_\psi^{(\Gamma)}} \pi_c(\theta)},
   \]
   \[
   l(\psi) = \frac{\sum_{\theta \in S_\psi^{(\Gamma)}} \pi_c(\theta)\cdot\mathbb{I}\!\left\{ i \preceq_\theta -i \text{ and } -i \not\preceq_\theta i \right\}}{\sum_{\theta \in S_\psi^{(\Gamma)}} \pi_c(\theta)}.
   \] 
   
   It is easy to verify that \(w(\psi) + t(\psi) + l(\psi) = 1\): for any final phase signal \(\theta \in \Theta^{(\Gamma)}\), the indicator functions satisfy
   \[
   \mathbb{I}\!\left\{ i \not\preceq_\theta -i \text{ and } -i \preceq_\theta i \right\} + \mathbb{I}\!\left\{ i \preceq_\theta -i \text{ and } -i \preceq_\theta i \right\} + \mathbb{I}\!\left\{ i \preceq_\theta -i \text{ and } -i \not\preceq_\theta i \right\} = 1.
   \]
   Summing over all \(\theta \in S_\psi^{(\Gamma)}\) with weights \(\pi_c(\theta)\) and normalizing by \(\sum_{\theta \in S_\psi^{(\Gamma)}} \pi_c(\theta)\) confirms the identity.

\item \textbf{Range-Based Clustering}:  
   Given \(n\) target clusters (a user-specified parameter), EHS partitions \(\Psi_i^{(r)}\) into \(n\) contiguous equity ranges. For example:
   \begin{itemize}
       \item Cluster 1: \(\text{equity}(\psi) \in [0, 1/n]\)
       \item Cluster 2: \(\text{equity}(\psi) \in (1/n, 2/n]\)\\
       \(\vdots\)
       \item Cluster \(n\): \(\text{equity}(\psi) \in ((n-1)/n, 1]\)
   \end{itemize}

   Signal observation infosets with identical equity values are guaranteed to fall into the same cluster.
\end{enumerate}

\begin{proof}[Proof of Proposition~\ref{prop:less-ehs}]

To prove PAOI is a resolution bound of EHS for a hold'em game modeled as a SOOG \(\mathcal{G}\), we need to show: for any two signal observation infosets \(\psi, \psi' \in \Psi_i^{(r)}\) in the same PAOI class (i.e., \(\psi, \psi' \in pi_j^{(r)}\) for some \(j\)), \(\psi\) and \(\psi'\) are assigned to the same EHS cluster for any number of target clusters \(n\).  

EHS clusters by contiguous equity ranges, so this holds if \(\psi\) and \(\psi'\) have identical \(\text{equity}\). To confirm this equity identity, we first show \(w(\psi) = w(\psi')\) and \(t(\psi) = t(\psi')\) for \(\psi, \psi' \in pi_j^{(r)}\)—a result that directly implies \(\text{equity}(\psi) = \text{equity}(\psi')\) (since \(\text{equity} = w + \frac{1}{2}t\)). We explicitly prove \(w(\psi) = w(\psi')\) below, and the proof for \(t(\psi) = t(\psi')\) follows the exact same logic.

\paragraph{Base Case: Final Phase \(r = \Gamma\)}
For the final phase (\(r = \Gamma\)), any infoset \(\psi \in \Psi_i^{(\Gamma)}\) is a set of terminal signals itself. By the definition of extended signals (Equation~\eqref{eq:extended-signals}), this means \(S_\psi^{(\Gamma)} = \psi\).  

Combining this with EHS's winning rate definition, we directly link \(w(\psi)\) to PAOI's winrate outcome feature component:  
\[
w(\psi) = \frac{\sum_{\theta \in S_\psi^{(\Gamma)}} \pi_c(\theta)\cdot\mathbb{I}\!\left\{ i \not\preceq_\theta -i \text{ and } -i \preceq_\theta i \right\}}{\sum_{\theta \in S_\psi^{(\Gamma)}} \pi_c(\theta)} = \frac{\sum_{\theta \in \psi} \pi_c(\theta)\cdot\mathbb{I}\!\left\{ i \not\preceq_\theta -i \text{ and } -i \preceq_\theta i \right\}}{\sum_{\theta \in \psi} \pi_c(\theta)} = wo_i^{(\Gamma),2}(\psi).
\]  

For \(\psi, \psi' \in pi_j^{(\Gamma)}\) (same PAOI class in the final phase), their PAOFs satisfy \(pf^{(\Gamma)}(\psi) = pf^{(\Gamma)}(\psi')\) (by definition of PAOI classes). This implies their winrate components are equal: \(wo_i^{(\Gamma),2}(\psi) = wo_i^{(\Gamma),2}(\psi')\). Thus:  
\[
w(\psi) = wo_i^{(\Gamma),2}(\psi) = wo_i^{(\Gamma),2}(\psi') = w(\psi').
\]  

The base case holds.

\paragraph{Inductive Step: Assume Hold for Phase \(r+1\), Prove for Phase \(r\)}
\subparagraph{Inductive Hypothesis: }
For phase \(r+1\), any two infosets \(\psi^1, \psi^2 \in \Psi_i^{(r+1)}\) in the same PAOI class (i.e., \(\psi^1, \psi^2 \in pi_k^{(r+1)}\) for some \(k\)) satisfy \(w(\psi^1) = w(\psi^2)\). We denote this common winning rate for all infosets in \(pi_k^{(r+1)}\) as \(w(pi_k^{(r+1)})\).  

\subparagraph{To Prove: }
For phase \(r\), any two infosets \(\psi, \psi' \in \Psi_i^{(r)}\) in the same PAOI class (i.e., \(\psi, \psi' \in pi_j^{(r)}\) for some \(j\)) satisfy \(w(\psi) = w(\psi')\).

We use the perfect-recall property of SOOGs and PAOI class characteristics to link \(w(\psi)\) to the PAOF \(pf_j^{(r)}(\psi)\):  
\begin{align*}
w(\psi) 
&= \frac{\sum_{\theta \in S_\psi^{(\Gamma)}} \pi_c(\theta)\cdot\mathbb{I}\!\left\{ i \not\preceq_\theta -i \text{ and } -i \preceq_\theta i \right\}}{\sum_{\theta \in S_\psi^{(\Gamma)}} \pi_c(\theta)}
\\
&= \frac{\sum_{\theta \in S_\psi^{(\Gamma)}} \pi_c(\theta)\cdot\mathbb{I}\!\left\{ i \not\preceq_\theta -i \text{ and } -i \preceq_\theta i \right\} \cdot \sum_{\psi' \in \Psi_i^{(r+1)}} \mathbb{I}\!\left\{ \theta \in S_{\psi'}^{(\Gamma)} \right\}}{\sum_{\theta \in S_\psi^{(\Gamma)}} \pi_c(\theta)}
  \quad \text{(any } \theta \in \Theta^{(\Gamma)} \text{ belongs to exactly one } S_{\psi'}^{(\Gamma)}, \text{ so the inner sum = 1)}
\\
&= \frac{\sum_{\psi' \in \Psi_i^{(r+1)}} \sum_{\theta \in S_\psi^{(\Gamma)}} \pi_c(\theta)\cdot\mathbb{I}\!\left\{ \theta \in S_{\psi'}^{(\Gamma)} \right\} \mathbb{I}\!\left\{ i \not\preceq_\theta -i \text{ and } -i \preceq_\theta i \right\}}{\sum_{\theta \in S_\psi^{(\Gamma)}} \pi_c(\theta)}
  \quad \text{(swap summation order for finite signal sets)}
\\
&= \frac{\sum_{\psi' \in \Psi_i^{(r+1)}} \sum_{\theta \in S_{\psi}^{(\Gamma)} \cap S_{\psi'}^{(\Gamma)}} \pi_c(\theta)\cdot\mathbb{I}\!\left\{ i \not\preceq_\theta -i \text{ and } -i \preceq_\theta i \right\}}{\sum_{\theta \in S_\psi^{(\Gamma)}} \pi_c(\theta)}
  \quad \text{(restrict to } \theta \text{ in the intersection of } S_\psi^{(\Gamma)} \text{ and } S_{\psi'}^{(\Gamma)} \text{)}
\\
&= \frac{\sum\limits_{j}^{\mathcal{C}^{(r+1)}} \sum_{\psi' \in pi_j^{(r+1)}} \sum_{\theta \in S_{\psi}^{(\Gamma)} \cap S_{\psi'}^{(\Gamma)}} \pi_c(\theta)\cdot\mathbb{I}\!\left\{ i \not\preceq_\theta -i \text{ and } -i \preceq_\theta i \right\}}{\sum_{\theta \in S_\psi^{(\Gamma)}} \pi_c(\theta)}
  \quad \text{(group by phase-\((r+1)\) PAOI class } pi_j^{(r+1)} \text{)}
\\
&= \frac{\sum\limits_{j}^{\mathcal{C}^{(r+1)}} \left( \sum_{\psi' \in pi_j^{(r+1)}} \left( \sum_{\theta \in S_{\psi'}^{(\Gamma)}} \pi_c(\theta) \right) \cdot w(\psi') \right)}{\sum_{\theta \in S_\psi^{(\Gamma)}} \pi_c(\theta)}
  \quad \text{(perfect-recall implies } S_\psi^{(\Gamma)} \cap S_{\psi'}^{(\Gamma)} = S_{\psi'}^{(\Gamma)} \text{ or } \emptyset; \text{ rearrange } w(\psi') = \frac{\sum \pi_c(\theta)\cdot\mathbb{I}\{\cdot\}}{\sum \pi_c(\theta)} \text{)}
\\
&= \sum\limits_{j}^{\mathcal{C}^{(r+1)}} \underbrace{\frac{\sum_{\psi' \in pi_j^{(r+1)} \cap S_\psi^{(r+1)}} \sum_{\theta \in S_{\psi'}^{(\Gamma)}} \pi_c(\theta)}{\sum_{\theta \in S_\psi^{(\Gamma)}} \pi_c(\theta)}}_{pf_j^{(r)}(\psi)} \cdot w(pi_j^{(r+1)})
\end{align*}

The key equivalence here relies on two critical properties of hold'em games. First, every signal \(\theta\) in the same phase has an identical number of successor signals across all future phases. For example, each pre-flop signal in Texas Hold'em branches into the same number of flop signals. Second, the transition probabilities \(\xi(\theta, \theta')\) are uniform, meaning \(\xi(\theta, \theta') = 1/M\) for any \(\theta'\) that is a valid successor of \(\theta\), where \(M\) is the fixed number of successors per signal (guaranteed by the first property). These properties ensure that the total weighted final descendants of a signal set factor into a product of the sum of \(\pi_c(\theta)\) for signals \(\theta \in \psi\) and a constant term. This constant term, which accounts for uniform transitions through all future phases, is invariant across signals and thus cancels out in the numerator and denominator of the fraction. As a result, the fraction reduces to exactly the revised PAOF definition:  

\[ pf_j^{(r)}(\psi) = \frac{\sum_{\theta \in \psi} \pi_c(\theta) \cdot \sum_{\psi' \in pi^{(r+1)}_j}\sum_{\theta' \in \psi'} \xi(\theta, \theta')}{\sum_{\theta \in \psi} \pi_c(\theta)\cdot \sum_{\theta' \in \Theta^{(r+1)}} \xi(\theta, \theta')}. \]  

Additionally, by the inductive hypothesis, all \(\psi' \in pi_j^{(r+1)}\) share the same winning rate \(w(pi_j^{(r+1)})\), justifying the replacement of \(w(\psi')\) with the class-level rate.

Since \(\psi, \psi' \in pi_j^{(r)}\) (same PAOI class), their PAOFs satisfy \(pf_j^{(r)}(\psi) = pf_j^{(r)}(\psi')\) for all \(j\); meanwhile, \(w(pi_j^{(r+1)})\) is a constant for each phase-\((r+1)\) PAOI class \(j\). Thus:  
\[
w(\psi) = \sum_{j} pf_j^{(r)}(\psi) \cdot w(pi_j^{(r+1)}) = \sum_{j} pf_j^{(r)}(\psi') \cdot w(pi_j^{(r+1)}) = w(\psi').
\]  

The inductive step holds.

\paragraph{Extension to \(\text{equity}(\psi) = \text{equity}(\psi')\)}
As noted earlier, the proof for \(t(\psi) = t(\psi')\) (tie rate) is identical to that of \(w(\psi) = w(\psi')\): replacing the winning indicator \(\mathbb{I}\{i \not\preceq_\theta -i \text{ and } -i \preceq_\theta i\}\) with the tie indicator \(\mathbb{I}\{i \preceq_\theta -i \text{ and } -i \preceq_\theta i\}\) leads to \(t(\psi) = t(\psi')\) via the same base case and inductive reasoning.  

Substituting into EHS's equity formula gives:  
\[
\text{equity}(\psi) = w(\psi) + \frac{1}{2}t(\psi) = w(\psi') + \frac{1}{2}t(\psi') = \text{equity}(\psi').
\]

\paragraph{Conclusion}
All signal observation infosets in the same PAOI class have identical \(\text{equity}\) for all phases \(r\) in the hold'em SOOG \(\mathcal{G}\), so they are assigned to the same EHS cluster for any \(n\). Thus, the PAOI abstraction serves as a resolution bound of EHS.  

Proposition~\ref{prop:less-ehs} is proven.
\end{proof}

\subsection{Proof of Proposition~\ref{prop:less-paa}}
\label{subsec:proof-less-paa}
Proposition~\ref{prop:less-paa} states: For a hold'em game modeled as a SOOG \(\mathcal{G}\), the PAOI abstraction serves as a resolution bound of algorithm PAAEMD. We first clarify the core logic of the PAAEMD algorithm (with reused EHS content clearly marked), then proceed with the proof.
\paragraph{PAAEMD Algorithm}
PAAEMD (Potential-Aware Abstraction with Earth Mover's Distance) is a state-of-the-art lossy hand abstraction algorithm for 2-player hold'em games. It reuses EHS's equity calculation for final phase initialization and adopts EHS-consistent L2-based clustering for final phase cluster partitioning, while extending with phase-wise transition histograms and EMD for non-final phase clustering. Its key steps (adapted to the SOOG framework) are as follows:
\begin{enumerate}[left=0cm]
\item \textbf{Final Phase Equity Calculation (Reused from EHS)}:
PAAEMD inherits EHS's expected showdown equity definition for final phase (\(r = \Gamma\)): for any signal observation infoset \(\psi \in \Psi_i^{(\Gamma)}\), \(\text{equity}(\psi) = w(\psi) + \frac{1}{2}t(\psi) + 0 \cdot l(\psi)\)  (consistent with EHS, including the definition of $w(\psi), t(\psi), l(\psi)$; see Section~\ref{subsec:proof-less-ehs} for details).
\item \textbf{Final Phase L2-Based Clustering}:
Using the equity values from Step (1), PAAEMD applies \(k\)-means clustering to partition \(\Psi_i^{(\Gamma)}\) into \(m^{(\Gamma)}\) clusters (denoted \(c^{(\Gamma)}_1, c^{(\Gamma)}_2, \dots, c^{(\Gamma)}_{m^{(\Gamma)}}\)). The distance metric is squared L2 distance (Euclidean distance squared) between equity values.
\item \textbf{Non-Final Phase Histogram Construction and EMD Clustering}:
For non-final phases (\(r < \Gamma\)), PAAEMD links to phase \(r+1\) clusters (from Step (2) or recursive Step (3)) via two sub-steps:\\
1) \textit{Transition Histogram Construction}: For any signal observation infoset \(\psi \in \Psi_i^{(r)}\), build a histogram \(h(\psi) \in [0,1]^{m^{(r+1)}}\) where the \(j\)-th element is the conditional probability of transitioning to phase \(r+1\) cluster \(c^{(r+1)}_j\):\[
   h_j(\psi) = \frac{\sum\limits_{\theta \in \psi} \pi_c(\theta) \cdot \sum\limits_{\psi' \in c^{(r+1)}_j}\sum\limits_{\theta' \in \psi'} \xi(\theta, \theta')}{\sum_{\theta \in \psi} \pi_c(\theta)\cdot \sum\limits_{\theta' \in \Theta^{(r+1)}} \xi(\theta, \theta')}.
   \]
2) \textit{EMD-Based Clustering}: Apply \(k\)-means to partition \(\Psi_i^{(r)}\) into \(m^{(r)}\) clusters (denoted \(c^{(r)}_1, c^{(r)}_2, \dots, c^{(r)}_{m^{(r)}}\)), using Earth Mover's Distance (EMD) between histograms as the distance metric.
\end{enumerate}
\begin{proof}[Proof of Proposition~\ref{prop:less-paa}]
To prove PAOI is a resolution bound of PAAEMD for a hold'em game modeled as a SOOG \(\mathcal{G}\), we need to show: for any two signal observation infosets \(\psi, \psi' \in \Psi_i^{(r)}\) in the same PAOI class (i.e., \(\psi, \psi' \in pi_j^{(r)}\) for some \(j\)), \(\psi\) and \(\psi'\) are assigned to the same PAAEMD cluster for any number of target clusters \(m^{(r)}\).
PAAEMD's clustering depends on equity (final phase) or transition histograms (non-final phases), so this holds if: (1) final phase: \(\text{equity}(\psi) = \text{equity}(\psi')\); (2) non-final phases: \(h(\psi) = h(\psi')\) (implying EMD = 0). We reuse EHS's result for (1) and prove (2) via mathematical induction.
\paragraph{Base Case: Final Phase \(r = \Gamma\)}
From the proof of Proposition~\ref{prop:less-ehs} (Base Case), we already established:
\textit{Any two infosets \(\psi, \psi' \in \Psi_i^{(\Gamma)}\) in the same PAOI class (i.e., \(pf^{(\Gamma)}(\psi) = pf^{(\Gamma)}(\psi')\)) have identical expected showdown equity (\(\text{equity}(\psi) = \text{equity}(\psi')\))}.
PAAEMD uses squared L2 distance for final phase clustering—since \(\text{equity}(\psi) = \text{equity}(\psi')\), their distance to all cluster centroids is identical, so they are assigned to the same cluster. The base case holds.
\paragraph{Inductive Step: Assume Hold for Phase \(r+1\), Prove for Phase \(r\)}
\subparagraph{Inductive Hypothesis:}
For phase \(r+1\), any two infosets \(\psi^1, \psi^2 \in \Psi_i^{(r+1)}\) in the same PAOI class (i.e., \(\psi^1, \psi^2 \in pi_k^{(r+1)}\) for some \(k\)) satisfy \(h(\psi^1) = h(\psi^2)\) (if \(r+1 < \Gamma\)) or \(\text{equity}(\psi^1) = \text{equity}(\psi^2)\) (if \(r+1 = \Gamma\)). We denote this common histogram (or equity) for all infosets in \(pi_k^{(r+1)}\) as \(h(pi_k^{(r+1)})\) (or \(\text{equity}(pi_k^{(r+1)})\)).
\subparagraph{To Prove:}
For phase \(r\), any two infosets \(\psi, \psi' \in \Psi_i^{(r)}\) in the same PAOI class (i.e., \(\psi, \psi' \in pi_j^{(r)}\) for some \(j\)) satisfy \(h(\psi) = h(\psi')\).
\subparagraph{Derivation of \(h(\psi)\) via PAOF}
Starting from PAAEMD's transition histogram definition (Algorithm Step (3)), we use the perfect-recall property of SOOGs and PAOI class characteristics to link \(h(\psi)\) to the PAOF \(pf_j^{(r)}(\psi)\) (consistent with EHS's PAOF reasoning in Section~\ref{subsec:proof-less-ehs}):
\begin{align*}
h_j(\psi)
&= \frac{\sum\limits_{\theta \in \psi} \pi_c(\theta) \cdot \sum\limits_{\psi' \in c^{(r+1)}_j}\sum\limits_{\theta' \in \psi'} \xi(\theta, \theta')}{\sum\limits_{\theta \in \psi} \pi_c(\theta)\cdot \sum\limits_{\theta' \in \Theta^{(r+1)}} \xi(\theta, \theta')}
\\
&= \frac{\sum\limits_{\theta \in \psi} \pi_c(\theta) \cdot \sum\limits_{k: pi_k^{(r+1)} \subseteq c^{(r+1)}_j} \sum\limits_{\psi' \in pi_k^{(r+1)}}\sum\limits_{\theta' \in \psi'} \xi(\theta, \theta')}{\sum\limits_{\theta \in \psi} \pi_c(\theta)\cdot \sum\limits_{\theta' \in \Theta^{(r+1)}} \xi(\theta, \theta')}
  \quad \text{(by inductive hypothesis: each } pi_k^{(r+1)} \text{ is exactly assigned to one } c^{(r+1)}_j \text{ for some } j\text{)}
\\
&= \sum\limits_{k: pi_k^{(r+1)} \subseteq c^{(r+1)}_j} \underbrace{\frac{\sum\limits_{\theta \in \psi} \pi_c(\theta) \cdot \sum\limits_{\psi' \in pi_k^{(r+1)}}\sum\limits_{\theta' \in \psi'} \xi(\theta, \theta')}{\sum\limits_{\theta \in \psi} \pi_c(\theta)\cdot \sum\limits_{\theta' \in \Theta^{(r+1)}} \xi(\theta, \theta')}}_{pf_k^{(r)}(\psi)}
\end{align*}

Since \(\psi, \psi' \in pi_j^{(r)}\) (same PAOI class), their PAOFs satisfy \(pf_k^{(r)}(\psi) = pf_k^{(r)}(\psi')\) for all \(k\); meanwhile, the grouping of \(k\) by \(c^{(r+1)}_j\) is a constant for phase \(r+1\). Thus:\(
h_j(\psi) = \sum_{k: pi_k^{(r+1)} \subseteq c^{(r+1)}_j} pf_k^{(r)}(\psi) = \sum_{k: pi_k^{(r+1)} \subseteq c^{(r+1)}_j} pf_k^{(r)}(\psi') = h_j(\psi').
\)

For all \(j\), \(h_j(\psi) = h_j(\psi')\), so \(h(\psi) = h(\psi')\). The inductive step holds.
\paragraph{Extension to EMD Clustering Consistency}
PAAEMD uses EMD between transition histograms for non-final phase clustering (Algorithm Step 3(2)). If \(h(\psi) = h(\psi')\), the EMD between them is zero—meaning \(\psi\) and \(\psi'\) have identical distance to all cluster centroids, so they are assigned to the same cluster for any \(m^{(r)}\).
\paragraph{Conclusion}
All signal observation infosets in the same PAOI class have identical equity (final phase, reused from EHS) and identical transition histograms (non-final phases, via induction) for all phases \(r\) in the hold'em SOOG \(\mathcal{G}\), so they are assigned to the same PAAEMD cluster. Thus, the PAOI abstraction serves as a resolution bound of PAAEMD.
Proposition~\ref{prop:less-paa} is proven.
\end{proof}

\end{document}